\documentclass[12pt,oneside,reqno]{amsart}
\usepackage{amsmath,amssymb,amsthm,textcomp}
\usepackage{amsfonts,graphicx}
\usepackage[mathscr]{eucal}
\usepackage{color}
\usepackage[table]{xcolor}
\usepackage{booktabs}
\usepackage{epstopdf}

\usepackage{subfig}
\usepackage{epsfig}
\usepackage{caption}
\usepackage{hyperref}
\theoremstyle{definition}
\usepackage{float}
\usepackage{multirow}
\floatstyle{plaintop}
\restylefloat{table}
\usepackage{lscape}
\numberwithin{equation}{section}

\newcommand{\ncom}{\newcommand}

\ncom{\beq}{\begin{equation}}
\ncom{\eeq}{\end{equation}}
\ncom{\bea}{\begin{eqnarray*}}
	\ncom{\eea}{\end{eqnarray*}}
\ncom{\beqa}{\begin{eqnarray}}
\ncom{\eeqa}{\end{eqnarray}}
\ncom{\nno}{\nonumber}
\ncom{\non}{\nonumber}
\ncom{\ds}{\displaystyle}
\ncom{\half}{\frac{1}{2}}
\ncom{\mbx}{\makebox{.25cm}}
\ncom{\hs}{\mbox{\hspace{.25cm}}}
\ncom{\rar}{\rightarrow}
\ncom{\Rar}{\Rightarrow}
\ncom{\noin}{\noindent}
\ncom{\bc}{\begin{center}}
	\ncom{\ec}{\end{center}}
\ncom{\sz}{\scriptsize}
\ncom{\rf}{\ref}
\ncom{\s}{\sqrt{2}}
\ncom{\sgm}{\sigma}
\ncom{\Sgm}{\Sigma}
\ncom{\psgm}{\sigma^{\prime}}
\ncom{\dt}{\delta}
\ncom{\Dt}{\Delta}
\ncom{\lmd}{\lambda}
\ncom{\Lmd}{\Lambda}
\ncom{\Th}{\Theta}
\ncom{\e}{\eta}
\ncom{\eps}{\epsilon}
\ncom{\pcc}{\stackrel{P}{>}}
\ncom{\lp}{\stackrel{L_{p}}{>}}
\ncom{\dist}{{\rm\,dist}}
\ncom{\sspan}{{\rm\,span}}
\ncom{\re}{{\rm Re\,}}
\ncom{\im}{{\rm Im\,}}
\ncom{\sgn}{{\rm sgn\,}}
\ncom{\ba}{\begin{array}}
	\ncom{\ea}{\end{array}}
\ncom{\hone}{\mbox{\hspace{1em}}}
\ncom{\htwo}{\mbox{\hspace{2em}}}
\ncom{\hthree}{\mbox{\hspace{3em}}}
\ncom{\hfour}{\mbox{\hspace{4em}}}
\ncom{\vone}{\vskip 2ex}
\ncom{\vtwo}{\vskip 4ex}
\ncom{\vonee}{\vskip 1.5ex}
\ncom{\vthree}{\vskip 6ex}
\ncom{\vfour}{\vspace*{8ex}}
\ncom{\norm}{\|\;\;\|}
\ncom{\integ}[4]{\int_{#1}^{#2}\,{#3}\,d{#4}}
\ncom{\vspan}[1]{{{\rm\,span}\{ #1 \}}}
\ncom{\dm}[1]{ {\displaystyle{#1} } }
\ncom{\ri}[1]{{#1} \index{#1}}

\newtheorem{theorem}{\bf Theorem}[section]
\newtheorem{remark}{\bf Remark}[section]
\newtheorem{proposition}{Proposition}[section]

\newtheorem{definition}{Definition}[section]
\newtheoremstyle
{remarkstyle}
{}
{11pt}
{}
{}
{\bfseries}
{:}
{     }
{\thmname{#1} \thmnumber{#2} }

\theoremstyle{remarkstyle}

\usepackage{setspace}

\def\eps{\varepsilon}

\begin{document}
	\title{\large  A n\lowercase{ovel}   b\lowercase{ivariate} g\lowercase{eneralized} w\lowercase{eibull}  d\lowercase{istribution} \lowercase{with} p\lowercase{roperties} \lowercase{and} a\lowercase{pplications}}
	\author[Ashok Kumar Pathak]{Ashok Kumar Pathak$^{1}$}
	\author{Mohd. Arshad$^{2, *}$}
	\author{Qazi J. Azhad$^{3}$}
	\author{Mukti Khetan$^{4}$}
	\author{Arvind Pandey$^{5}$	\\\\ 
		$^{1}$D\lowercase{epartment of} M\lowercase{athematics and} S\lowercase{tatistics}, C\lowercase{entral} U\lowercase{niversity of} P\lowercase{unjab},
		B\lowercase{athinda}, I\lowercase{ndia}.	
		\\$^{2}$D\lowercase{epartment of} M\lowercase{athematics}, I\lowercase{ndian} I\lowercase{nstitute of} T\lowercase{echnology}
		I\lowercase{ndore}, S\lowercase{imrol}, I\lowercase{ndore}, I\lowercase{ndia}.\\
		$^{3}$D\lowercase{epartment of} M\lowercase{athematics and} S\lowercase{tatistics}, B\lowercase{anasthali} v\lowercase{idyapith}, R\lowercase{ajasthan}, I\lowercase{ndia}.\\
			$^{4}$D\lowercase{epartment of} M\lowercase{athematics}, A\lowercase{mity}, U\lowercase{niversity} M\lowercase{umbai}, M\lowercase{aharashtra}, I\lowercase{ndia}.\\
				$^{5}$D\lowercase{epartment of} S\lowercase{tatistics}, C\lowercase{entral} U\lowercase{niversity} \lowercase{of} R\lowercase{ajasthan}, R\lowercase{ajasthan}, I\lowercase{ndia}.
	}

		\thanks{*Corresponding author. \\
			E-mail address: ashokiitb09@gmail.com (Ashok Kumar Pathak), 
			arshad.iitk@gmail.com (Mohd. Arshad),  qaziazhadjamal@gmail.com (Qazi J. Azhad), mukti.khetan11@gmail.com (Mukti Khetan), arvindmzu@gmail.com (Arvind Pandey).}

	\begin{abstract}
			{
		Univariate Weibull distribution is a well known lifetime distribution and has been widely used in reliability and survival analysis.
		In this paper, we introduce a new family of bivariate generalized Weibull (BGW) distributions, whose univariate marginals are exponentiated Weibull distribution. Different statistical quantiles like  marginals, conditional distribution, conditional expectation, product moments, correlation and a measure component reliability are derived. Various measures of dependence and statistical properties along with ageing properties are examined. Further, the copula associated with BGW distribution and its  various important properties are also considered.	The methods of maximum likelihood and Bayesian estimation are employed to estimate unknown parameters of the model. A Monte Carlo simulation and real data study are carried out to demonstrate the performance of the estimators and results have proven the effectiveness of the distribution in real-life situations.}
	\end{abstract}
	
	\maketitle
\section{Introduction} \noindent{ The Weibull distribution is a natural extension of exponential and Rayleigh distributions, and is extensively used for modeling lifetime data with constant, strictly increasing and decreasing hazard functions.  The cumulative distribution function of a two parameter Weibull random variable $U$ with parameters $a$ and $b$ (denoted by $U\sim W(a, b))$ is given by
	\begin{equation*}
	F_{U}(x)=1-e^{-\displaystyle b x^{a}},\;\;{x> 0},
	\end{equation*}
	where $a, b> 0$.
Several generalizations of the Weibull distribution have been proposed by introducing additional parameters (see for example, Mudholkar and Srivastva (1993), Xie {\it et al.} (2002), Bebbington {\it et al.} (2007), Alshangiti (2014), Almalki (2018), Park and Park (2018), Gen and Songjian (2019), Bahman and Mohammad (2021)). Generalized Weibull distribution does not only includes a large family of well know distributions, but also has a broader range of hazard rate functions, which enhance the flexibility of models in modeling complex lifetime data. These distributions have vast applications in diverse disciplines like reliability, environmental, social science and medicine.}
 \noindent{Distributions are key elements for modeling dependence among random  variables. 
 	Recently, the constructions of new bivariate distributions with specified marginals have received lots of attention for theoretical and practical purposes.   
 	Various new state-of-the-art techniques for constructing bivariate or multivariate distributions have been  discussed in the literature. Some of these important techniques include, cumulative hazard rate function, conditional distribution, order statistics and copula function (see Balakrishnan and Lai (2009), Sarabia and Emilio (2008), Samanthi and Sepanski (2019))}.\\
 {Marshall and Olkin (1967) presented a bivariate generalization of the exponential distribution having Weibull marginals. This distribution is well known as Marshall-Olkin bivariate Weibull (MOBW) distribution and is most commonly used in practical applications.  The MOBW distribution has absolutely continuous and singular components and is useful in competing risk modeling. Some of the important references include Lee (1979), Hanagal (1996),  Kundu and Gupta (2010), Nandi and Dewan (2010), and Jose {\it et al.} (2011).
 	 Lu and Bhattacharyya (1990) proposed a new bivariate Weibull (BW) distribution which can model both positive and negative dependence. Marshall and Olkin (1997) constructed a new family of bivariate Weibull distribution by adding a parameter in the Weibull model and established its various properties. Kundu and Gupta (2014) discussed a new five parameter flexible geometric-Weibull distribution, which is a generalization of the Weibull distributions. Recently, some new family of bivariate Weibull distributions also have been proposed and studied in the literature.  Al-Mutairi {\it et al.} (2018) proposed a new four-parameter bivariate weighted Weibull distribution whose joint probability density function can be either a decreasing or unimodal function. This model is useful in analyzing a wide class of bivariate data in practice.  Barbiero (2019) constructed a new bivariate distribution with discrete marginals via Farlie-Gumbel-Morgenstern copula and performed the Monte Carlo simulation study to demonstrate the performance of the different estimation techniques. Recently, Gongsin and Saporu (2020) derived a new bivariate distribution using conditional and marginal Weibull distributions and utilized this model in renewable energy data. Bai {\it et al.} (2020) discuss the inferential aspect of Marshall-Olkin bivariate Weibull distribution with application in competing risks.
 	 }
 	
{This paper aims to introduce a new absolutely continuous bivariate generalized Weibull (BGW) distribution, whose marginals are a member of exponentiated Weibull family of distributions. The proposed distribution has the bivariate generalized exponential (BGE) as a sub-model studied by Mirhosseini {\it et al.} (2015). The bivariate generalized Rayleigh (BGR) distribution discussed by Pathak and Vellaisamy (2020) is also a sub-model of the proposed BGW distribution. Several important properties of the BGE, BGR, and their mixtures can be easily studied on a common platform via BGW distribution. The proposed model can be utilized as a better alternative to BGE and BGR models in practical applications. Various statistical properties along with some concept of dependence are discussed for the proposed BGW distribution. We obtain the copula associated with BGW distribution and derive the various measures of dependence based on copula. The values of these measures are also plotted for different values of copula parameters.  With the help of copula, we demonstrate that the proposed distribution exhibits a strong positive dependence and can be useful in numerous real situations.}

 {The structure of the article is as follows: In Section 2, we introduced a new family of bivariate Weibull distribution and deduce some existing families of well known distributions and their extensions. In Section 3, we derive the expressions for joint density, conditional density and conditional distribution for the BGW distribution. We also obtain the expressions for product moments and distribution of minimum order statistics. Section 4, presents some concept of dependence and discuss ageing properties for the BGW family. In Section 5, we obtain the copula associated with BGW distribution and some measures of association in terms of copulas. Sector 6 deals with the methodology of maximum likelihood and Bayesian estimation to estimate unknown parameters of the model. Section 7 presents the detailed Monte Carlo simulation study to validate the performances of the  estimators. Section 8 discusses the application of real-data set and its interpretations; the paper ends with conclusions.}
 \section{Bivariate Generalized Weibull  Distribution}
 \noindent {Consider a sequence of independent Bernoulli trials in which the $i$-th trial has probability of success  $\theta/i$, $0<\theta\leq 1$, $i\in\{1,2,\ldots\}$.
 Let $K$ denote the trial number on which the first success occurs. Then the probability mass function and probability generating function of random variable $K$ is (see {Pathak and Vellaisamy (2020)} or {Dolati {\it et al.} (2014)} or {Mirhosseini {\it et al.} (2015)})
 \begin{align*}
 P(K=k)&=\left(1-\theta\right)\left(1-\frac{\theta}{2}\right)\ldots\left(1-\frac{\theta}{k-1}\right)\frac{\theta}{k}\nonumber\\
 &=\frac{(-1)^{k-1}}{k!}{\theta(\theta-1)\ldots(\theta-k+1)},
 \end{align*}
 for $k=1,2,\ldots$, and 
 \begin{equation}\label{PG1}
 h_{K}(s)=E\left(s^{K}\right)=1-(1-s)^{\theta},\;\;s\in[0,1],
 \end{equation}
 respectively.\\
 Consider that $\{U_1, U_2,\ldots\}$ and  $\{V_1, V_2,\ldots\}$ are two sequences of mutually independent and identically distributed (i.i.d.) random variables, where
 $U_{i}\sim \mathrm{W}(a, b_1)$ and $V_{i}\sim \mathrm{W}(a, b_2)$ for $i\in\{1,2,3,\ldots\}$.
 Define $X:=\min(U_1,\ldots,U_K)$ and $Y:=\min(V_1,\ldots,V_{K})$. The joint survival function of $(X,Y)$ is given by
{
 \begin{align}\label{survival1}
 S(x,y)&=P(X> x, Y> y)\nonumber\\
 &=P\left(\min(U_1,\ldots,U_K)> x, \min(V_1,\ldots,V_{K})> y\right)\nonumber\\
 &=\sum_{k=1}^{\infty}\left[P(U_{i}> x) P(V_{i}> y)\right]^{k}P[K=k]\nonumber\\
 &=h_{K}\left(e^{-\displaystyle(b_1 x^{a}+b_2 y^{a})}\right)\nonumber\\
 &=1-\left\{1-e^{-\displaystyle(b_1 x^{a}+b_2 y^{a})}\right\}^{\theta}.
 \end{align}}
\noindent A bivariate random vector $(X,Y)$ is said to have a bivariate generalized Weibull distribution with parameters $a, b_1, b_2$ and $\theta$,  if its joint distribution function is given by
 \begin{align}\label{dist1}
 F(x,y)=&\left\{1-e^{-\displaystyle b_1 x^{a}}\right\}^{\theta}+\left\{1-e^{-\displaystyle b_2 y^{a}}\right\}^{\theta}
 -\left\{1-e^{-\displaystyle(\displaystyle b_1 x^{a}+\displaystyle b_2 y^{a})}\right\}^{\theta},
 \end{align} 
 where {$x,y\geq 0$} and $a, b_{1}, b_{2} > 0$ and $0<\theta\leq 1$. It is denoted by BGW$(a, b_1, b_2,\theta)$.\\
  \noindent The joint probability density function of the BGW distribution 
  is given by
  {
 \begin{equation}\label{denst1}
 f(x,y)=\frac{\partial^{2}F(x,y)}{\partial x\partial y}=\theta a^2 b_{1}b_{2} x^{a-1}y^{a-1}e^{-Z(x,y;\psi)}\left(1-e^{-Z(x,y;\psi)}\right)^{\theta-2}
 \left(1-\theta e^{-Z(x,y;\psi)}\right),
 \end{equation}
 where $Z(x,y;\psi):=Z(x,y,a, b_{1},b_{2})=b_{1}x^{a}+b_{2}y^{a}$ and $\psi=(a, b_{1}, b_{2})$.}}\\
It may be observed that $X\sim\mathrm{EW}(a,b_{1}, \theta)$, which is a member of exponentiated Weibull (EW) distribution having distribution function $F_{X}(x)=\left\{1-e^{-b_1x^{a}}\right\}^{\theta}$, $x\geq 0$ (see Mudholkar and Srivastva (1993)). Also, generalized exponential distribution with parameters $b_1$ and $\theta$ {\it i.e.,} $X\sim \text{GE}(b_1,\theta)$ is a sub-model of EW model, when $a=1$ (see Gupta and Kundu (1999)). Similarly, $Y\sim\mathrm{EW}(a,b_{2}, \theta)$.\\
The BGW family includes a large class of well-known families of distributions and their extensions. Some important special cases of BGW distribution are as follows:
{
\begin{itemize}
	\item [(i)]  {\bf Bivariate Generalized Exponential Distribution:} When $a=1$, from (\ref{dist1}) the joint distribution of random vector $(X,Y)$ is
		\begin{equation}\label{dist11}
		F(x,y)=\left\{1-e^{-b_{1}x}\right\}^{\theta}+\left\{1-e^{-b_{2}y}\right\}^{\theta}-\left\{1-e^{-(b_{1}x+b_{2}y)}\right\}^{\theta},
		\end{equation}
		where $x, y\geq 0$, $b_{1}, b_{2}>0$, and $0<\theta \leq 1$, which is the bivariate generalized exponential (BGE) distribution proposed by Mirhosseini {\it et al.} (2015).
		\item [(ii)]{\bf Bivariate Generalized Rayleigh Distribution:} When $a=2$, we have from (\ref{dist1})
		{
			\begin{equation}\label{RD1}
			F(x,y)=\left\{1-e^{-b_1x^{2}}\right\}^{\theta}+\left\{1-e^{-b_2y^{2})}\right\}^{\theta}-\left\{1-e^{-\left(b_1x^{2}+b_2y^{2}\right)}\right\}^{\theta},
			\end{equation}
			where $x, y\geq 0$, $b_{1}, b_{2}>0$, and $0<\theta \leq 1$,
			which is a bivariate generalized Rayleigh (BGR) distribution with parameters $b_{1}, b_{2}$ and $\theta$ as discussed by Pathak and Vellaisamy (2020).}
		\item[(iii)]  For $\theta=1$, equation (\ref{dist1}) leads to independence of $X$ and $Y$ with distribution 
		\begin{equation*}
		F(x,y)=\left\{1-e^{-b_1x^{a}}\right\}\left\{1-e^{-b_2 y^{a}}\right\},
		\end{equation*}
		where $x, y\geq 0$, $a>0$ and $b_{i}>0$ for $i=1,2$.
\end{itemize}
Different surface plots of joint distribution and density of the BGW distribution, given in (\ref{dist1}) and (\ref{denst1}), are presented in Figure \ref{Fig1} for different parameter values.
   \begin{figure}[t]
   	\centering
   	\subfloat[]{\includegraphics[width=0.43\textwidth]{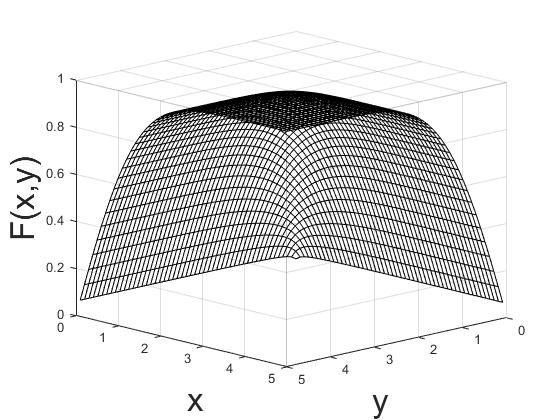}}              
   	\subfloat[]{\includegraphics[width=0.43\textwidth]{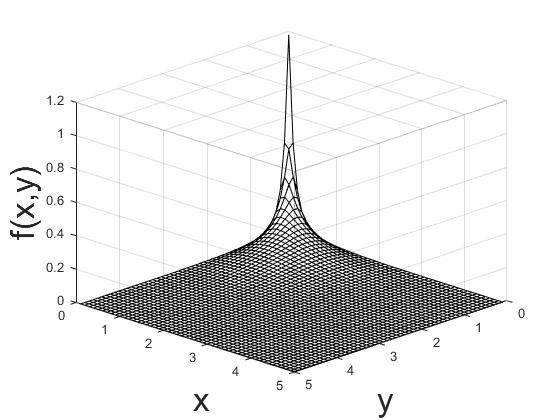}} \\
   	\subfloat[]{\includegraphics[width=0.43\textwidth]{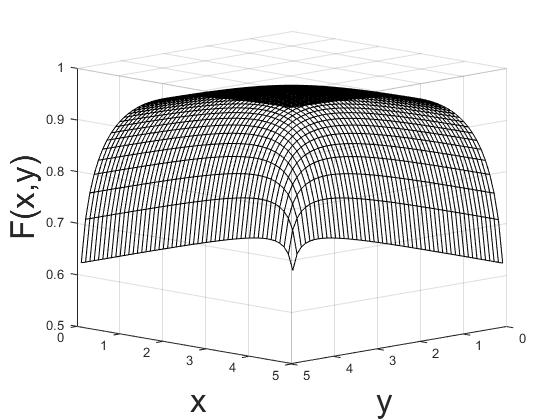}}              
   	\subfloat[]{\includegraphics[width=0.43\textwidth]{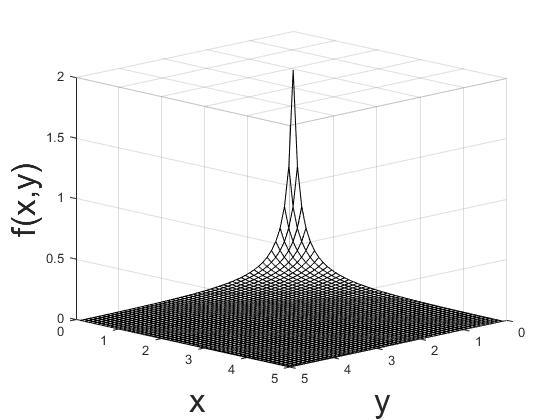}}
   	\caption{Surface plots of $F(x,y)$ and $f(x,y)$ of the BGW distribution: In (A) and (B), $a=1$, $b_{1}=0.5$, $b_{2}=0.5$, $\theta=0.5$. 
   		In (C) and (D), $a=1$, $b_{1}=1$, $b_{2}=1$, $\theta=0.2$.}\label{Fig1}
   \end{figure}
\noindent With the help of binomial series expansion of $\left(1-e^{-Z(x,y;\psi)}\right)^{\theta}$, the survival and density functions of the BGW distribution are
 \begin{equation*}
   S(x,y)=\sum_{j=1}^{\infty}\binom{\theta}{j}(-1)^{j+1} e^{-jZ(x,y;\psi)},
   \end{equation*}
  and 
  \begin{align}
  f(x,y)&=a^2 b_1 b_2x^{a-1}y^{a-1}\sum_{j=1}^{\infty}\binom{\theta}{j}(-1)^{j+1} j^{2}e^{-jZ(x,y;\psi)}\nonumber\\
  &=a^2b_1 b_2x^{a-1}y^{a-1}\sum_{j=1}^{\infty}\binom{\theta}{j}(-1)^{j+1} j^{2}e^{-j\left(b_1x^{a}+b_2y^{a}\right)}\label{dist3},
   \end{align}
   respectively.
   \section{Basic Properties}
   \noindent In this section, some basic quantities of BGW distribution such as condition density, conditional distribution function and conditional survival function will be derived. Distribution of minimum order statistic and  stress-strength reliability parameter are obtained. Expression for regression function for the BGW distribution and its sub-models will also be reported. We will also derive product moments and calculate the correlation coefficient for the BGW distribution.\\ 
  Using basic definitions, the following result is easy to establish. 
\begin{theorem}\label{Theorem1}
	Let $(X,Y)\sim \mathrm{BGW}(a,b_{1}, b_{2},\theta)$. Then
	\begin{itemize}
		\item [(i)] the conditional density function of $Y$ given $X=x$ is 
		\begin{equation*}
		f(y|x)= \frac{ab_2 y^{a-1} e^{-b_2y^{a}}\left(1-e^{-Z(x, y;\phi)}\right)^{\theta -2}\left(1-\theta e^{-Z(x, y;\psi)}\right)}{\left\{1-e^{-b_1 x^{a}}\right\}^{\theta-1}},
		\end{equation*}
		\item [(ii)] the conditional distribution of $Y$ given $X=x$ is 
		\begin{equation*}
		F(y|x)=P(Y\leq y|X=x)=1-\frac{e^{-b_2y^{a}}\left(1-e^{-Z(x, y;\psi)}\right)^{\theta-1}}{\left\{1-e^{-b_1x^{a}}\right\}^{\theta-1}},
		\end{equation*}
		\item [(iii)] the conditional survival function of $Y$ given $X=x$ is 
		\begin{equation*}
		S(y|x) =P(Y>y|X=x)=\frac{e^{-b_2y^{a}}\left(1-e^{-Z(x, y;\psi)}\right)^{\theta-1}}{\left\{1-e^{-b_1x^{a}}\right\}^{\theta-1}}.
		\end{equation*}
	\end{itemize}
\end{theorem}

	The following result gives the expression for regression function of BGW model.
	\begin{theorem}\label{regproposition}
		Let $(X,Y)$ be a bivariate random vector having BGW distribution. Then the regression function of $Y$ on $X=x$ is 
		\begin{equation}\label{reg00}
		E(Y|X=x)=\displaystyle\frac{a b_1\Gamma (1+1/a)x^{a-1}}{b_{2}^{1/a}f_{X}(x)} \displaystyle\sum_{j=1}^{\infty}\binom{\theta}{j}(-1)^{j+1} j^{1-1/a} e^{-jb_{1}x^{a}},
		\end{equation}
		where $f_{X}(x)$ is the marginal density of $X$.
	\end{theorem}
	\begin{proof}
		The proof is given in Appendix.
	\end{proof}
\noindent	From (\ref{reg00}), we can get regression function of several well known distributions studied in literature. In particular, for $a=1$, (\ref{reg00}) reduces to 
	\begin{equation*}
	E(Y|X=x)=\frac{b_{1}}{b_{2}f_{X}(x)}\left[1-\left(1-e^{-b_{1}x}\right)^{\theta}\right],
	\end{equation*}
	which has been established for the BGE distribution in Mirhosseini {\it et al.} (2015). \\

In the following result, we derive an expression for product moments for the BGW distribution and from it we deduce product moments for some known families of distributions. Also, we calculate the coefficient of correlation for the BGW families of distributions. 
\begin{theorem}\label{Momenttheorem}
	Let $(X,Y)\sim \mathrm{BGW}(a,b_{1}, b_{2},\theta)$. Then
		\begin{equation}\label{moment1}
	E(X^{r}Y^{s})=\displaystyle\frac{\Gamma(1+r/a)\Gamma(1+s/a)}{b_{1}^{r/a}b_{2}^{s/a}}\sum_{j=1}^{\infty}\binom{\theta}{j}(-1)^{j+1}\frac{1}{j^{(r+s)/a}}.
	\end{equation}
\end{theorem}
\begin{proof}
	Proof is given in Appendix.
\end{proof}
From Theorem \ref{Momenttheorem}, we have the following:
\begin{itemize}
	\item [(i)] For $a=1$, (\ref{moment1}) yields
	\begin{equation}\label{moment2}
	E(X^{r}Y^{s})=\displaystyle r!~s!\sum_{j=1}^{\infty}\binom{\theta}{j}\frac{(-1)^{j+1}}{b_{1}^{r} b_{2}^{s}j^{r+s}},
	\end{equation}
	the product moments of the BGE distribution discussed in Mirhosseini {\it et al.} (2015). Specially, for $r=s=1$, (\ref{moment2}) gives
		\begin{equation*}
	E(XY)=\sum_{j=1}^{\infty}\binom{\theta}{j}\frac{(-1)^{j+1}}{b_{1} b_{2} j^{2}},
	\end{equation*}
	which has been considered in Mirhosseini {\it et al.} (2015).
	\item [(ii)] When $a=2$ and $r=s=1$, from (\ref{moment1}) the product moment of the BGR defined in (\ref{RD1}),  is
		\begin{equation*}
	E(XY)=\frac{\pi}{4} \sum_{j=1}^{\infty}\binom{\theta}{j}\frac{(-1)^{j+1}}{\sqrt{b_{1}b_{2}}~j}.
	\end{equation*}
\end{itemize}
{
 {Now consider $X\sim \text{EW}(a, b_1,\theta)$, then its  $r$th moment about the origin is denoted by $A(a,b_1,\theta,r)$ and is given by 
 	\begin{equation}\label{moment3}
 A(a,b_1,\theta,r)=E(X^{r})=\displaystyle \frac{\theta\Gamma(1+r/a)}{b_1^{r/a}}\sum_{j=0}^{\infty}\binom{\theta-1}{j}\frac{(-1)^{j}}{(j+1)^{1+r/a}}.
 \end{equation}}
 A similar expression $B(a,b_2,\theta,s)$ for the $s$th moment  for $Y$ can also be obtained.\\  For $r=s=1$ and with the help of (\ref{moment1}) and (\ref{moment3}) through a simple algebra, the coefficient of correlation for the BGW distribution is given by
 \begin{equation*}
 R(X,Y)=\frac{\displaystyle\frac{\Gamma(1+1/a)\Gamma(1+1/a)}{b_{1}^{1/a}b_{2}^{1/a}}\sum_{j=1}^{\infty}\binom{\theta}{j}(-1)^{j+1}\frac{1}{j^{2/a}}-A(a,b_1,\theta,1)B(a,b_2,\theta,1)}{\sqrt{A(a,b_1,\theta,2)-A^2(a,b_1,\theta,1)}\sqrt{B(a,b_2,\theta,2)-B^2(a,b_2,\theta,1)}}.
 \end{equation*}
 For $\theta=1$, $R(X,Y)=0$ which corresponds to the independence of $X$ and $Y$.}\\
 \noindent In the next result, we derive expressions for the distribution  of minimum order statistic and stress-strength parameter for the BGW distribution. 
 \begin{theorem}\label{orderstat}
 	If $(X,Y)\sim \mathrm{BGW}(a,b_{1}, b_{2},\theta)$, then
 	\begin{itemize}
 		\item [(i)] $\min(X,Y)\sim \mathrm{EW}(a,b_{1}+b_{2},\theta)$
 		\item [(ii)] $P(X<Y)=\displaystyle\frac{b_2}{b_1+b_2}.$
 	\end{itemize}
 \end{theorem}
 \begin{proof}
 	Proof is given in Appendix A.
 \end{proof}
 \noindent {Let $X$ and $Y$ be the lifetimes of two components in a system. Then $\min(X,Y)$ may be observed as the lifetimes of two components series system.  System will work as long as both components functioning together. It may be applicable in measuring the reliability of computer networking, electronic circuits {\it etc.}
 	\begin{remark}
 		It may be notice that for $b_1=b_2=b$ (say), $P(X<Y)=1/2$.
 	\end{remark}
In the forthcoming sections, we discuss some measures of the local dependences for the BGW distribution and discuss its important properties.
{
\section{Dependence and Ageing Properties} \noindent
The notion of dependence among random variables is very useful in reliability theory and lifetime data analysis. Covariance and product moment correlation are classical techniques for measuring the strength of dependence between two variables. Apart from these classical measures, several other notions of new dependence have been proposed in the literature. In this section, we study various dependence properties namely, positive quadrant dependence, regression dependence, stochastic increasing, totally positivity of order 2, {\it etc.}  of the proposed BGW distribution. Furthermore, we also study some ageing properties of the BGW  under different bivariate ageing definitions. First, we proceed with positive quadrant dependence.
\begin{definition}
	Let $(X,Y)$ be a bivariate random vector with distribution and marginals $F(x,y)$, $F_X(x)$ and $F_Y(y)$, respectively. We say that $(X,Y)$ is positive quadrant dependent (PQD) if 
	\begin{equation*}
	F(x,y)\geq F_X(x)F_Y(y)\;\; \text{for~all~ $x$~ and~ $y$},
	\end{equation*}
	or, equivalently, if
	\begin{equation*}
	S(x,y)\geq S_X(x)S_Y(y)\;\; \text{for~all~ $x$~ and~ $y$,}
	\end{equation*}
	where $S(x,y)$,  $S_X(x)$ and $S_Y(y)$ denotes the joint and marginals survival functions. The random vector $(X,Y)$ is negative quadrant dependent (NQD) if reverse inequality holds (see Lehmann (1966) and Nelsen (2006)).
\end{definition}
\begin{proposition}
	Let $(X,Y)$ follows  $\mathrm{BGW}(a,b_{1}, b_{2},\theta)$. Then $(X,Y)$ is PQD.
\end{proposition}
\begin{proof}
From (\ref{survival1}), one can easily get marginal survival functions $S_X(x)$ and $S_Y(y)$. With the help of joint and marginal survival function, one can easily establish that $S(x,y)\geq S_X(x)S_Y(y)$, which corresponds to the PQD of the BGW distribution.
\end{proof} 
\begin{remark}
	$X$ and $Y$ are positively correlated if $\text{Cov}(X,Y)\geq 0$. Hence, a direct consequences of PQD property, leads to $\text{Cov}(X,Y)\geq 0$, for the BGW family.
\end{remark}
Regression dependence is stronger concept of dependence than PQD. Here, we study the measure of regression dependence for the BGW distribution.
\begin{definition}
$F(x,y)$ is positively regression dependent if (see Nelsen (2006))
\begin{equation*}
P(Y>y|X=x)~ \text{is ~increaing ~in~ }x~ \text{for ~all~ values~ of}~ y.
\end{equation*}
\end{definition} 
\begin{proposition}
	Let $(X,Y)$ follows the BGW distribution with distribution function $F(x,y)$. Then $F(x,y)$ in (\ref{dist1}) is positively regression dependent.
\end{proposition}
\begin{proof}
The conditional survival function $P(Y>y|X=x)$ of $Y$ on $X=x$ is reported in (iii) point of the Theorem (\ref{Theorem1}). On differentiation with respect to $x$, we get
\begin{equation*}
\frac{\partial}{\partial x}P(Y>y|X=x)=(\theta-1)ab_1x^{a-1}(e^{-b_2 y^{a}}-1)(1-e^{-b_1 x^{a}})^{\theta-2} e^{-Z(x,y;\psi)}\left(1- e^{-Z(x,y;\psi)}\right)^{\theta-2}\geq 0.
\end{equation*}
This completes the proof of result.
\end{proof}}
{
We next review some other basic definitions related to dependence. A details discussion on these dependence can be found in Nelsen (2006).
\begin{definition}
	$Y$ is left tail decreasing in $X$ (denoted as LTD($Y|X$)) if $P(Y\leq y|X\leq x)$ is a nonincreasing function in $x$ for all $y$.
\end{definition}
\begin{definition}
	The random vector $(X,Y)$ is said to be left corner set decreasing (LCSD) if $P(X\leq x, Y\leq y|X\leq x_{1}, Y\leq y_{1})$ is nonincreasing in $x_{1}$ and $y_{1}$ for all $x$ and $y$.\
\end{definition}
\begin{proposition}\label{LTD}
	Let $(X,Y)\sim \mathrm{BGW}(a,b_{1}, b_{2},\theta)$. Then
	\begin{itemize}
		\item [(i)] $(X,Y)$ is LTD.
		\item [(ii)] $(X,Y)$ is LCSD.
	\end{itemize}
\end{proposition}
\noindent To prove the Proposition \ref{LTD}, it suffices to establish  the totally positivity of order 2 (TP2) of density $f$, which is a strongest concept of dependence. As TP2 is equivalent to LCSD and implies to LTD (see Nelsen (2006), and Balakrishnan and Lai (2009)).\\

\noindent  In order to establish the TP2 property of the BGW distribution, we  begin with a local dependence function.
To study the dependence between random variables $X$ and $Y$, Holland and Wang (1987) proposed a local dependence function $\delta(x,y)$ as
\begin{equation*}
\delta (x,y)=\frac{\partial^{2}}{\partial x \partial y}\ln f(x,y).
\end{equation*}
This dependence function provides a powerful tool to study the TP2 property of a bivariate distribution. Some detailed properties of the $\delta (x,y)$ have been studied in Holland and Wang (1987) and Balakrishnan and Lai (2009).
 \begin{proposition}\label{propdep1}
	Let $(X,Y)\sim \mathrm{BGW}(a,b_{1}, b_{2},\theta)$. Then
	\begin{align*}
	\delta(x,y)&=a^2b_1b_2 x^{a-1}y^{b-1} e^{-Z(x, y;\psi)}\left[\frac{(2-\theta)}{(1-e^{-Z(x, y;\psi)})^{2}}-\frac{\theta}{(1-\theta e^{-Z(x, y;\psi)})^{2}}\right].
	\end{align*}
\end{proposition}
\noindent It may notice that, when $\theta=1$, then $\delta (x,y)=0$, which leads to the independence of $X$ and $Y$.\\
Holland and Wang (1987) established that a bivariate density $f(x,y)$ will possess the TP2 property if and only if $\delta (x,y)\geq 0$.\\
Now, we have the following result:
\begin{theorem}\label{TP2}
	Let $(X, Y)\sim \text{BGW}(a,b_1, b_2, \theta)$. Then, for $0<\theta\leq 1$, the density $f(x,y)$ given in (\ref{denst1}) is TP2.
\end{theorem}
\noindent Let $(X, Y)$ be a bivariate random vector with joint density $f(x,y)$ and survival function $S(x,y)$. Then, the bivariate hazard rate function is defined as (see {Basu (1971)}) 
\begin{equation}\label{hr1}
h(x,y)=\displaystyle\frac{f(x,y)}{S(x,y)}.
\end{equation}
If  $(X, Y)\sim \text{BGW}(a,b_1, b_2, \theta)$, then we have
\begin{equation*}
h(x,y)=\frac{\theta a^2 b_{1}b_{2} x^{a-1}y^{a-1}e^{-Z(x,y;\psi)}\left(1-e^{-Z(x,y;\psi)}\right)^{\theta-2}
	\left(1-\theta e^{-Z(x,y;\psi)}\right)}{\left[1-\left\{1-e^{-Z(x,y;\psi)}\right\}^{\theta}\right]}.
\end{equation*}
If $\theta=1$, $h(x,y)$ leads to product of  two marginal failure rate
functions.
{
\subsection{Hazard gradient functions} The hazard components of a bivariate random vector $(X,Y)$ are defined as (see Johnson and Kotz (1975))
	\begin{equation*}
	\eta_1(x,y)=-\frac{\partial}{\partial x}\ln S(x,y)
	\end{equation*}
	and 
	\begin{equation*}
	\eta_2(x,y)=-\frac{\partial}{\partial y}\ln S(x,y).
	\end{equation*}
	The vector $(\eta_1(x,y), \eta_2(x,y))$ are termed as the hazard gradient of a bivariate random vector $(X, Y)$. It may notice that $\eta_1(x,y)$ is  conditional hazard rate of $X$ given information $Y>y$ and  $\eta_2(x,y)$ is conditional hazard rate of $Y$ given information $X>x$.\\ Hence, for the BGW distribution the hazard gradient is 
	 \begin{equation}\label{r1}
	\eta_{1}(x,y)=	\frac{\theta ab_1x^{a-1}e^{-Z(x, y;\psi)}\left(1-e^{-Z(x, y;\psi)}\right)^{\theta-1}}{\left\{1-\left(1-e^{-Z(x, y;\psi)}\right)^{\theta}\right\}},
	\end{equation}
	and
	\begin{equation}\label{r2}
	\eta_{2}(x,y)=\frac{\theta ab_2y^{a-1}e^{-Z(x, y;\psi)}\left(1-e^{-Z(x, y;\psi)}\right)^{\theta-1}}{\left\{1-\left(1-e^{-Z(x, y;\psi)}\right)^{\theta}\right\}}.
	\end{equation}
	Next result demonstrates the monotonicity of the conditional hazard rate functions.
	\begin{proposition}
	Let  $(X, Y)\sim \text{BGW}(a,b_1, b_2, \theta)$. Then 
	\begin{itemize}
		\item [(i)] $\eta_{1}(x,y)$ is deceasing in $y$.
		\item [(ii)] $\eta_{2}(x,y)$ is deceasing in $x$.
	\end{itemize}
	\end{proposition}
\begin{proof}
	Due to Shaked (1977), if $f(x,y)$ is TP2, then conditional hazard rate $\eta_{1}(x,y)$ is deceasing in $y$ and $\eta_{2}(x,y)$ is deceasing in $x$. Hence, by virtue of TP2 property of BGW family and Shaked (1977) results, proof is immediate.
\end{proof}}
	\begin{proposition}
		The BGW distribution in (\ref{dist1}) is bivariate decreasing hazard rate (DHR).
	\end{proposition}
\section{Copulas and dependence measures}\noindent The dependencies between two random variables $X$ and $Y$ are completely determined by its joint distribution $F(x,y)$. Copula is a powerful tool to study the dependence between variables. Any distribution function can be expressed in the form of copula, in which dependence and marginals can be studied separately. Sklar (1959) showed that any joint distribution function $F$ can be expressed in the form
\begin{equation}\label{copula1}
F(x,y)=C(F_{X}(x), F_{Y}(y))\;\;\text{for~all}~ x, y\in \mathbb{R}.
\end{equation}
For continuous $F_X$ and $F_Y$, the representation (\ref{copula1}) is unique. In discrete case, it is uniquely determined on the $\text{Range}(F_X)\times \text{Range}(F_Y)$.\\ Let $F_{X}^{-1}$ and $F_{Y}^{-1}$ be the inverse distribution functions of continuous random variables $X$ and $Y$, respectively. Then, for every $s,t\in [0,1] $, one can easily obtain the copula $C$ as follows:
\begin{equation*}
C(s,t)=F(F_{X}^{-1}(s),F_{Y}^{-1}(t)).
\end{equation*}
Let $(X,Y)$ have the BGW distribution. Then associated copula is given by
\begin{eqnarray}\label{copula3}
C(s,t)&=s+t-\left\{1-\left(1-s^{\frac{1}{\theta}}\right)\left(1-t^{\frac{1}{\theta}}\right)\right\}^{\theta}\nonumber\\
&=s+t-st\left\{s^{-\frac{1}{\theta}}-t^{-\frac{1}{\theta}}-1\right\}^{\theta}.
\end{eqnarray}
It may be notice that the copula $C$ associated with the BGW family is the same as the copula reported in Mirhosseini {\it et al.} (2015) and Pathak and Vellaisamy (2020) for the bivariate generalized exponential (BGE) distribution  and bivariate generalized linear exponential (BGLE) distribution, respectively.\\
The product moments correlation is a measure of linear dependence and  may give misleading results even in the case of strong dependence for non-elliptical random variables. As the copulas are invariant under the monotonic transformation of random variables. Therefore, the copula based measures of concordance are capable to capture non-linear dependence and are usually considered as the best alternative to linear correlation. First of all, we consider some important measures of dependence based on copulas for the BGW family, namely Spearman's rho ($\rho$), Kendall's tau ($\tau$), Blest's measure ($B$), and Spearman's footrule coefficients ($\phi$). For  definitions and important properties, once may refer to Nelsen (1998, 2006) and Genest and Plante (2003).\\
The following result is due to Dolati {\it et al.} (2014), Mirhosseini {\it et al.} (2015), and Pathak and Vellaisamy (2020).
	\begin{proposition}
	For the $\text{BGW}(a,b_1, b_2, \theta)$ family
	\begin{align*}
	\rho(X,Y)&=9-12\theta^{2}\sum_{j=0}^{\infty} (-1)^{j}\binom{\theta}{j}\big[B(\theta,j+1)\big]^{2},\\
	\tau(X,Y)&=1+4\theta B(2,2\theta+1)\big(\Psi(2)-\Psi(2\theta+1)\big),\\
	\phi(X,Y)&=4-6\theta\sum_{j=0}^{\infty}\binom{\theta}{j}(-1)^{j}B(\theta,2j+1),
	\end{align*} 
	and
	\begin{align*}
	B(X,Y)=8-24\theta^{2}\sum_{j=0}^{\infty}\binom{\theta}{j}(-1)^{j}B(\theta, j+1)\big[B(\theta, j+1)-B(2\theta, j+1)\big],
	\end{align*}
	where $B(a,b)$ is beta function and $\Psi$ denotes the digamma function defined as
	$\Psi=\frac{d}{du}\ln \Gamma(u)$, where $\Gamma(u)$ is the gamma function.
\end{proposition}
Now, we have the following interconnection between Spearman's rho and Kendall's tau for the BGW family. 
\begin{theorem}
	If $(X,Y)$ follows the $\mathrm{BGW}(a,b_{1}, b_{2},\theta)$, the $\rho(X,Y)$ and $\tau(X,Y)$ are non-negative and  $\rho(X,Y) >\tau(X,Y)$. 
\end{theorem}

\begin{proof}
	For $0<\theta \leq 1$, Theorem \ref{TP2} shows that BGW family is TP2. Therefore, $\rho(X,Y)$ and $\tau(X,Y)$ are non-negative. TP2 property implies that $X$ and $Y$ are positively quadrant dependent. By an exercise of Proposition 2.3 of Cap\'{e}ra\'{a} and Genset (1993), we obtain that $\rho(X,Y) >\tau(X,Y)$.
\end{proof}
Next, we calculate tail dependence coefficients and derive the expression for measure of regression dependence for the copula associated with BGW distribution. 
\subsection{Tail dependence coefficient}
Tail dependence coefficients, evaluate the amount of dependence
on the tails of a joint bivariate distribution and can describe the extremal dependence. Let $C$ be a copula associated with a bivariate random vector $(X,Y)$. Then the coefficients of lower-tail dependence ($\lambda_L(C)$) and upper-tail dependence ($\lambda_U(C)$) are defined as (see Nelsen (2006), p. 214)
\begin{equation*}
\lambda_L(C)=\lim_{t\rightarrow 0^{+}}\frac{C(t,t)}{t},
\end{equation*}
and 
\begin{equation*}
\lambda_U(C)=\lim_{t\rightarrow 1^{-}}\frac{1-2t+C(t,t)}{1-t}.
\end{equation*}
The range of tail dependences is between 0 to 1. If $\lambda_L(C)>0$, then $X$ and $Y$ have lower-tail dependence and  if $\lambda_L(C)=0$, then no lower-tail dependence. Similarly, $\lambda_U(C)$ can also be interpreted.
For BGW family,
\begin{align*}
\lambda_L(C)=&\lim_{t\rightarrow 0^{+}}\frac{2t-t\{2-t^{\frac{1}{\theta}}\}^{\theta}}{t}=2-2^{\theta},
\end{align*}
and
\begin{equation*}
\lambda_U(C)=\lim_{t\rightarrow 1^{-}}\frac{1-\left\{1-\left(1-t^{\frac{1}{\theta}}\right)\right\}^{\theta}}{1-t}=0.
\end{equation*}
Hence, the BGW family  have lower-tail dependence but no upper-tail dependence.
\subsection{A measure of regression dependence} A measure of regression dependence between two random variables $X$
and $Y$ in terms of copula $C$ is defined as (see Dette {\it et al.} (2013))
\begin{equation}\label{regdep0}
r(X,Y)=6\int_{0}^{1}\int_{0}^{1}\left(\frac{\partial}{\partial s}C(s,t)\right)^2 ds~ dt-2.
\end{equation}
The range $r(X,Y)$ is in $[0,1]$. $r(X,Y)=1$ if and only if $Y=h(X)$ for some Borel measurable function $h$, and $r(X,Y)=0$ if and only if $X$ and $Y$ are independent.
{
\begin{theorem}\label{regressiondependence}
	Let $X$ and $Y$ be bivariate random variables with distribution belonging to the family of $\text{BGW}(a,b_1, b_2, \theta)$. Then
	\begin{equation*}
	r(X,Y)=4+6\theta^2\sum_{j=0}^{\infty}\binom{\theta-1}{j}{(-1)^{j}}\left[
	{B(2-\theta,j+1)}B(\theta,j+3)-2B(1,j+1)B(\theta,j+2)\right].
	\end{equation*}
\end{theorem}
\begin{proof}
	Appendix is given in Appendix.
	\end{proof}}
\begin{figure}[ht]
	\centering
	{\includegraphics[scale=.85]{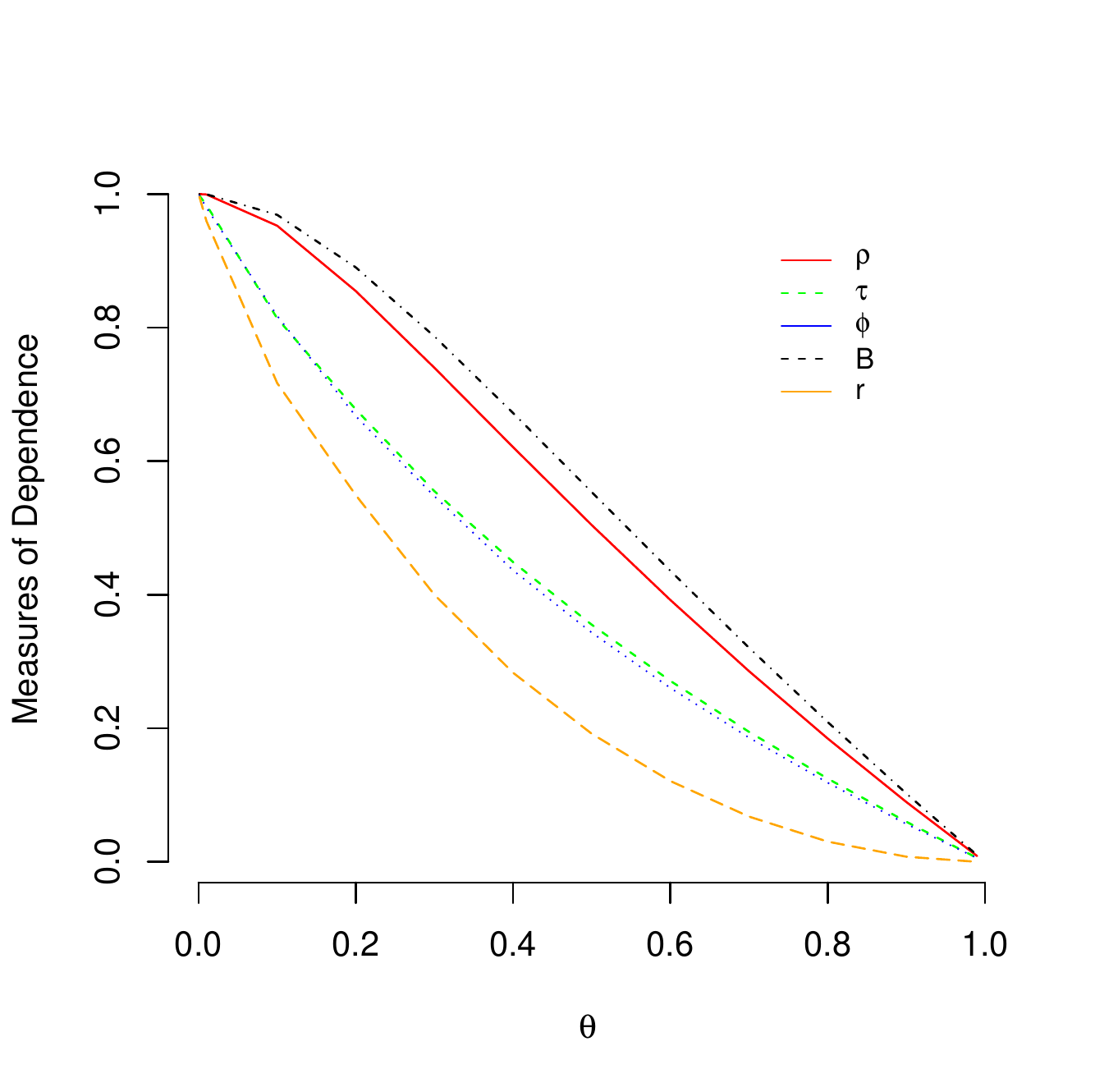}}
	\caption{Copula based measures of dependence for different parameter values. }\label{copulameasure1}
\end{figure}
We plot the numerical values of $\rho(X,Y)$, $\tau(X,Y)$, $\phi(X,Y)$, $B(X,Y)$, and $r(X,Y)$ for different values of copula parameter in Figure \ref{copulameasure1} to demonstrate the dependence structure. From Figure \ref{copulameasure1}, we see that these measures exhibit non-negative values, which correspond to the PQD of the copula.  Also, as the parameter $\theta$ tends to 1, the values of these measures approach to zero, which supports the independence of $X$ and $Y$.
\section{Estimation of parameter}

\noindent In this section, we consider the problem of estimation of unknown parameters  $a$, $b_1$, $b_2$ and $\theta$ for the BGW distribution using maximum likelihood and Bayesian approach.  First, we obtain the maximum likelihood estimates (MLEs) of the unknown parameters. 
\subsection{Maximum Likelihood Estimation}
 Let $\{(x_1,y_1),(x_2,y_2),\ldots,(x_n,y_n)\}$ be a sample of size $n$ from BGW($a,b_1,b_2,\theta$) distribution. The likelihood function based on this sample and density function given in (\ref{denst1}) is defined as 
$$L(\Theta|\boldsymbol{x},\boldsymbol{y})=\prod_{i=1}^{n}f(x_i,y_i;\psi),$$
where $\boldsymbol{x}=(x_1,x_2,\ldots,x_n)$ and $\boldsymbol{y}=(y_1,y_2,\ldots,y_n)$ are realizations of $\boldsymbol{X}$ and $\boldsymbol{Y}$, respectively, and $\Theta=(a,b_1,b_2,\theta).$
Now, the log-likelihood function is defined as
\begin{align}\label{log-likelihood}
\ln L(\Theta|\boldsymbol{x},\boldsymbol{y})=&n\left(\ln \theta+2\ln a+\ln b_1+\ln b_2\right)+(a-1)\sum_{i=1}^{n} \left(\ln x_i + \ln y_i\right)- \sum_{i=1}^{n} Z(x_i,y_i;\psi)
\nonumber\\	
&+(\theta-2)\sum_{i=1}^{n} \ln\left(1-e^{-Z(x_i,y_i;\psi)}\right)+\sum_{i=1}^{n} \ln\left(1-\theta e^{-Z(x_i,y_i;\psi)}\right).
\end{align}
In order to find the MLEs of $\Theta=(a,b_1,b_2,\theta),$ we differentiate \eqref{log-likelihood} with respect to $a,b_1,b_2,\theta$ and equate them to 0. The normal equations after differentiation \eqref{log-likelihood}, are given as
\begin{align*}
&\dfrac{2n}{a}+\sum_{i=1}^{n}(\ln x_i+\ln y_i)-\sum_{i=1}^{n}\left(b_1 x_i^a \ln x_i+b_2 y_i^a \ln y_i\right)+(\theta-2)\sum_{i=1}^{n}\dfrac{e^{-Z(x_i,y_i;\psi)} \left(b_1 x_i^a \ln x_i+b_2 y_i^a \ln y_i\right)}{1-e^{-Z(x_i,y_i;\psi)}}\\
&\qquad\quad+\theta\sum_{i=1}^{n}\dfrac{e^{-Z_(x_i,y_i;\psi)} \left(b_1 x_i^a \ln x_i+b_2 y_i^a \ln y_i\right)}{1-\theta e^{-Z_(x_i,y_i;\psi)}}=0\\
&\dfrac{n}{b_1}-\sum_{i=1}^n x_i^a+(\theta-2)\sum_{i=1}^n\dfrac{x_i^a e^{-Z(x_i,y_i;\psi)}}{1-e^{-Z(x_i,y_i;\psi)}}+\theta\sum_{i=1}^n\dfrac{x_i^a e^{-Z(x_i,y_i;\psi)}}{1-\theta e^{-Z(x_i,y_i;\psi)}}=0\\
&\dfrac{n}{b_2}-\sum_{i=1}^n y_i^a+(\theta-2)\sum_{i=1}^n\dfrac{y_i^a e^{-Z(x_i,y_i;\psi)}}{1-e^{-Z(x_i,y_i;\psi)}}+\theta\sum_{i=1}^n\dfrac{y_i^a e^{-Z(x_i,y_i;\psi)}}{1-\theta e^{-Z(x_i,y_i;\psi)}}=0\\	
&\dfrac{n}{\theta}+\sum_{i=1}^n \ln (1-e^{-Z(x_i,y_i;\psi)})+\sum_{i=1}^n\dfrac{e^{-Z(x_i,y_i;\psi)}}{1-\theta e^{-Z(x_i,y_i;\psi)}}=0.\\
\end{align*}
We see that normal equations are complex in nature and the manual solution of these equations is very tedious and quite cumbersome. So, we tend to computational aid to find out the MLEs of unknown parameters.
	\subsection{Bayesian Estimation}\label{mcmc}
In this section, we will obtain Bayes estimators of the unknown quantities of BGW distribution. For this purpose, we consider independent gamma priors for parameters $a$, $b_1$, $b_2$  i.e., $\pi(a)\sim Gamma(\delta_1,\zeta_1),$ $\pi(b_1)\sim Gamma(\delta_2,\zeta_2),$ $\pi(b_2)\sim Gamma(\delta_3,\zeta_3)$ and beta prior for $\theta$ i.e., $\pi(\theta)\sim Beta(\delta_4,\zeta_4).$ The joint posterior distribution of $\Theta$ is given as
\begin{equation}\label{first.posterior}
\pi(\Theta|\boldsymbol{x},\boldsymbol{y})\propto L(\Theta|\boldsymbol{x},\boldsymbol{y})\pi(\Theta).
\end{equation}
Now, according to our problem, equation \eqref{first.posterior} reduces to
\begin{align}\label{posterior}
\pi(\Theta|\boldsymbol{x},\boldsymbol{y})\propto& 
\theta^{n+\delta_4-1} a^{2n+\delta_1-1} b_1^{n+\delta_2-1}b_2^{n+\delta_3-1}  e^{-\zeta_1 a} e^{-\zeta_2 b_1} e^{-\zeta_3 b_2} (1-\theta)^{\zeta_4-1} \nonumber\\
&\prod_{i=1}^{n} \left( (x_i y_i)^{a-1}e^{-Z(x_i,y_i;\psi)}(1-e^{-Z(x_i,y_i;\psi)})^{\theta-2}(1-\theta e^{-Z(x_i,y_i;\psi)})\right), \nonumber\\
&\qquad\qquad	\qquad\qquad \qquad\qquad a>0,~b_1>0,~b_2>0,~\theta \in (0,1].
\end{align}
Further, we consider an asymmetric loss function called general entropy loss function {\it i.e.,} 
\begin{equation*}\label{GEloss}
l(\delta,\lambda)\propto \left(\dfrac{\delta}{\lambda}\right)^c-c\ln\left(\dfrac{\delta}{\lambda}\right)-1,\qquad c\ne  0
\end{equation*}
with corresponding Bayes estimator as
\begin{equation*}\label{GEbayes}
\delta_{GE}=	\left[E(\lambda^{-c})\right]^{-1/c}.
\end{equation*}
We see joint posterior density defined in \eqref{posterior} has a complex nature and finding out its expected value is again tedious. So, manually, it is quite impossible to obtain the Bayes estimators of the unknown quantities. But, we can employ the Markov chain Monte Carlo (MCMC) technique to find the approximate Bayes estimates with the aid of marginal posterior densities. The marginal posterior densities are calculated as
\begin{align*}
\pi(a|b_1,b_2,\theta,\boldsymbol{x},\boldsymbol{y})\propto&a^{2n+\delta_1-1}e^{-\zeta_1 a}\\
&\prod_{i=1}^{n} \left( (x_i y_i)^{a-1}e^{-Z(x_i,y_i;\psi)}(1-e^{-Z(x_i,y_i;\psi)})^{\theta-2}(1-\theta e^{-Z(x_i,y_i;\psi)})\right),\\
\pi(b_1|a,b_2,\theta,\boldsymbol{x},\boldsymbol{y})\propto&b_1^{n+\delta_2-1}e^{-\zeta_2 b_1}\\	
&\prod_{i=1}^{n} \left( e^{-Z(x_i,y_i;\psi)}(1-e^{-Z(x_i,y_i;\psi)})^{\theta-2}(1-\theta e^{-Z(x_i,y_i;\psi)})\right),\\
\pi(b_2|a,b_1,\theta,\boldsymbol{x},\boldsymbol{y})\propto&b_2^{n+\delta_3-1}e^{-\zeta_3 b_2}\\	
&\prod_{i=1}^{n} \left( e^{-Z(x_i,y_i;\psi)}(1-e^{-Z(x_i,y_i;\psi)})^{\theta-2}(1-\theta e^{-Z(x_i,y_i;\psi)})\right),\\
\pi(\theta|a,b_1,b_2,\boldsymbol{x},\boldsymbol{y})\propto&\theta^{n+\delta_4-1}\left(1-\theta\right)^{\zeta_4-1} \prod_{i=1}^{n} \left( (1-e^{-Z(x_i,y_i;\psi)})^{\theta-2}(1-\theta e^{-Z(x_i,y_i;\psi)})\right).
\end{align*}
We see that the marginal posterior densities of parameters do  not acquire any closed form of known distribution, so, generation of random samples from these densities is not simple. To tackle this situation, we employ the technique of MCMC with the aid of Metropolis-Hasting algorithm (See Gelman et al. (2013), Arshad {\it et al.} (2021), and Azhad {\it et al.} (2021)). 
\begin{description}
	\item[(i)] Initiate with prefixed value of $(a,b_1,b_2,\theta)$ as $(a^{0},b_1^{0},b_2^0,\theta^0).$
	\item[(ii)] Set j=1.
	\item[(iii)] Generate $a^{j},$ $b_1^{j},$ $b_2^{j}$  and $\theta^{j}$  from their respective marginal posterior densities given in Section \ref{mcmc} by employing Metropolis-Hasting algorithm and using initial values given in step (i).
	\item[(iv)] Repeat (ii)-(iii) for $j=1,2,\ldots \ldots,T$ times and obtain the generated samples of  $a,$ $b_1,$ $b_2,$ and $\theta.$
\end{description}
Now, the Bayes estimator, $\delta_{BE},$ can be found by using the following result
\begin{equation*}\label{mcmc.bayes.GE}
\delta_{BE}=\left(\dfrac{1}{T-N}\sum_{j=1}^{T-N}(\delta^j)^{-c}\right)^{-1/c},
\end{equation*}
where $N$ is the burn-in period. 
\section{Simulation Study}
In this section, a simulation study is conducted to exhibit the performances of derived various estimators under the paradigm of classical and Bayesian.  We have obtained  maximum likelihood estimators and MCMC Bayes estimators for unknown quantities. The performances of these estimators are measured based on the criteria of mean squared errors (MSE). In addition to that we have also provided the biases of the estimators. To obtain the MSEs and biases, we employ the Monte Carlo technique. The process is repeated 1000 times to observe the behaviour of estimators. These results are calculated for different configurations of the parameters and sample sizes. We have used the R software (R Core Team (2020)) for the calculation of the results. The results are calculated and reported in the Tables [\ref{tab:mcmc_1.5}-\ref{tab:mle}]. Table [\ref{tab:mcmc_1.5}] shows the biases and MSEs of Bayes estimates of the parameters $a$, $b_1,$ $b_2,$ $\theta$ for $(\delta_i,\zeta_i)=(1.5,1.5),$ $i=1,2,3,4$, $n=\{10,20,30,40\},$ and $c=\{0.5,1\}.$ Table [\ref{tab:mcmc_2}] shows the biases and MSEs of Bayes estimates of the parameters $a$, $b_1,$ $b_2,$ $\theta$ for $(\delta_i,\zeta_i)=(2,2)$,  $i=1,2,3,4$, $n=\{10,20,30,40\},$ and $c=\{0.5,1\}.$ The Markov chain is run for 10,000 times with the burn in period of 2000. Table [\ref{tab:mle}] represents the biases and MSEs of  maximum likelihood estimates of the parameters $a$, $b_1,$ $b_2,$ $\theta$ for different configurations. From all these tables, we observe that biases can not be used to observe the performances of the estimators as their behaviour is not consistent for all estimates. Whereas, we observe that MSEs are exhibiting a better picture for the performances of estimators. So, from Table [\ref{tab:mcmc_1.5}-\ref{tab:mle}], we conclude that MCMC Bayes estimates are performing better than MLE in most of the scenarios.
Also, the behaviour of generated samples using MCMC is depicted in Figures [\ref{fig:tracea}-\ref{fig:tracetheta}]. These figures exhibits trace plot of each generated sample of unknown quantity. 
\begin{landscape}
	\begin{table}[h]
		\centering
		\caption{Bias and mean squared error (MSE) of Bayes estimators for $(\delta_i,\zeta_i)=(1.5,1.5)$ for $i=1,2,3,4$.}

		\begin{tabular}{c|cccc|cccccccc}
			\toprule
			\multirow{2}[4]{*}{$n$} & \multicolumn{4}{c|}{\multirow{2}[4]{*}{$(a,b_1,b_2,\theta)$}} & \multicolumn{4}{c}{Bias}      & \multicolumn{4}{c}{MSE} \\
			\cmidrule{6-13}          & \multicolumn{4}{c|}{}         & $a$   & $b_1$ & $b_2$ & $\theta$ & $a$   & $b_1$ & $b_2$ & $\theta$ \\
			\midrule
			\multicolumn{13}{c}{$c=0.5$} \\
			\midrule
			10    & \multicolumn{4}{c|}{\multirow{4}[2]{*}{(2,1.5,1.5,0.5)}} & 0.5429 & 0.3750 & 0.3853 & 0.1057 & 0.4742 & 0.2207 & 0.2008 & 0.0169 \\
			20    & \multicolumn{4}{c|}{}         & 0.5389 & 0.3525 & 0.3237 & 0.0974 & 0.4710 & 0.2056 & 0.1505 & 0.0142 \\
			30    & \multicolumn{4}{c|}{}         & 0.5326 & 0.3363 & 0.2993 & 0.0903 & 0.4520 & 0.1926 & 0.1340 & 0.0126 \\
			40    & \multicolumn{4}{c|}{}         & 0.5094 & 0.3020 & 0.2767 & 0.0820 & 0.3841 & 0.1555 & 0.1150 & 0.0102 \\
			\midrule
			10    & \multicolumn{4}{c|}{\multirow{4}[2]{*}{(2,2,1.5,0.5)}} & 0.5317 & 0.4573 & 0.4398 & 0.1079 & 0.5331 & 0.3147 & 0.2537 & 0.0176 \\
			20    & \multicolumn{4}{c|}{}         & 0.5285 & 0.4249 & 0.3486 & 0.0998 & 0.4633 & 0.3081 & 0.1733 & 0.0148 \\
			30    & \multicolumn{4}{c|}{}         & 0.5202 & 0.4269 & 0.3155 & 0.0932 & 0.4398 & 0.3022 & 0.1485 & 0.0133 \\
			40    & \multicolumn{4}{c|}{}         & 0.4837 & 0.4034 & 0.2889 & 0.0847 & 0.3985 & 0.2840 & 0.1243 & 0.0110 \\
			\midrule
			10    & \multicolumn{4}{c|}{\multirow{4}[2]{*}{(2.5,2,1.5,0.5)}} & 0.6779 & 0.4692 & 0.3722 & 0.1016 & 0.7165 & 0.3106 & 0.1964 & 0.0154 \\
			20    & \multicolumn{4}{c|}{}         & 0.6676 & 0.4446 & 0.3299 & 0.1018 & 0.6456 & 0.3142 & 0.1609 & 0.0150 \\
			30    & \multicolumn{4}{c|}{}         & 0.6317 & 0.4480 & 0.2993 & 0.0975 & 0.6445 & 0.3283 & 0.1324 & 0.0133 \\
			40    & \multicolumn{4}{c|}{}         & 0.4744 & 0.4242 & 0.2743 & 0.0925 & 0.4562 & 0.3165 & 0.1141 & 0.0121 \\
			\midrule
			\multicolumn{13}{c}{$c=1$} \\
			\midrule
			10    & \multicolumn{4}{c|}{\multirow{4}[2]{*}{(2,1.5,1.5,0.5)}} & 0.5275 & 0.3810 & 0.4343 & 0.1114 & 0.4513 & 0.2089 & 0.2473 & 0.0185 \\
			20    & \multicolumn{4}{c|}{}         & 0.5266 & 0.3468 & 0.3447 & 0.1008 & 0.4381 & 0.2088 & 0.1670 & 0.0150 \\
			30    & \multicolumn{4}{c|}{}         & 0.5236 & 0.3312 & 0.3108 & 0.0930 & 0.4257 & 0.1844 & 0.1418 & 0.0132 \\
			40    & \multicolumn{4}{c|}{}         & 0.4813 & 0.2984 & 0.2837 & 0.0844 & 0.3734 & 0.1504 & 0.1196 & 0.0106 \\
			\midrule
			10    & \multicolumn{4}{c|}{\multirow{4}[2]{*}{(2,2,1.5,0.5)}} & 0.5225 & 0.4894 & 0.4982 & 0.1145 & 0.4987 & 0.3375 & 0.3138 & 0.0193 \\
			20    & \multicolumn{4}{c|}{}         & 0.5168 & 0.4194 & 0.3752 & 0.1034 & 0.4487 & 0.3007 & 0.1960 & 0.0157 \\
			30    & \multicolumn{4}{c|}{}         & 0.5040 & 0.4185 & 0.3307 & 0.0963 & 0.4180 & 0.2930 & 0.1603 & 0.0140 \\
			40    & \multicolumn{4}{c|}{}         & 0.4583 & 0.3969 & 0.2980 & 0.0872 & 0.3875 & 0.2712 & 0.1313 & 0.0115 \\
			\midrule
			10    & \multicolumn{4}{c|}{\multirow{4}[2]{*}{(2.5,2,1.5,0.5)}} & 0.6664 & 0.4977 & 0.4213 & 0.1066 & 0.6850 & 0.3427 & 0.2417 & 0.0168 \\
			20    & \multicolumn{4}{c|}{}         & 0.6524 & 0.4410 & 0.3485 & 0.1057 & 0.6282 & 0.3111 & 0.1755 & 0.0159 \\
			30    & \multicolumn{4}{c|}{}         & 0.6114 & 0.4396 & 0.3091 & 0.1006 & 0.6199 & 0.3033 & 0.1393 & 0.0140 \\
			40    & \multicolumn{4}{c|}{}         & 0.4508 & 0.4167 & 0.2808 & 0.0953 & 0.4184 & 0.3024 & 0.1182 & 0.0126 \\
			\bottomrule
		\end{tabular}%
		
		\label{tab:mcmc_1.5}%
	\end{table}%
\end{landscape}
\begin{landscape}
	\begin{table}[h]
		\centering
		\caption{Bias and mean squared error (MSE) of Bayes estimators for $(\delta_i,\zeta_i)=(2,2)$ for $i=1,2,3,4$.}

		\begin{tabular}{c|cccc|cccccccc}
			\toprule
			\multirow{2}[4]{*}{$n$} & \multicolumn{4}{c|}{\multirow{2}[4]{*}{$(a,b_1,b_2,\theta)$}} & \multicolumn{4}{c}{Bias}      & \multicolumn{4}{c}{MSE} \\
			\cmidrule{6-13}          & \multicolumn{4}{c|}{}         & $a$   & $b_1$ & $b_2$ & $\theta$ & $a$   & $b_1$ & $b_2$ & $\theta$ \\
			\midrule
			\multicolumn{13}{c}{$c=0.5$} \\
			\midrule
			10    & \multicolumn{4}{c|}{\multirow{4}[2]{*}{(2,1.5,1.5,0.5)}} & 0.5078 & 0.3330 & 0.3743 & 0.0928 & 0.4086 & 0.1642 & 0.1862 & 0.0137 \\
			20    & \multicolumn{4}{c|}{}         & 0.4963 & 0.3129 & 0.3095 & 0.0891 & 0.4078 & 0.1616 & 0.1385 & 0.0118 \\
			30    & \multicolumn{4}{c|}{}         & 0.4960 & 0.3043 & 0.2805 & 0.0835 & 0.3643 & 0.1539 & 0.1145 & 0.0107 \\
			40    & \multicolumn{4}{c|}{}         & 0.4290 & 0.2799 & 0.2626 & 0.0789 & 0.3468 & 0.1317 & 0.1019 & 0.0095 \\
			\midrule  10    & \multicolumn{4}{c|}{\multirow{4}[2]{*}{(2,2,1.5,0.5)}} & 0.4989 & 0.4647 & 0.4358 & 0.0959 & 0.3750 & 0.3046 & 0.2405 & 0.0141 \\
			20    & \multicolumn{4}{c|}{}         & 0.4890 & 0.3655 & 0.3405 & 0.0930 & 0.3703 & 0.2257 & 0.1605 & 0.0132 \\
			30    & \multicolumn{4}{c|}{}         & 0.4684 & 0.3747 & 0.3042 & 0.0873 & 0.3390 & 0.2180 & 0.1319 & 0.0114 \\
			40    & \multicolumn{4}{c|}{}         & 0.3887 & 0.3609 & 0.2734 & 0.0810 & 0.3161 & 0.2120 & 0.1081 & 0.0098 \\
			\midrule    10    & \multicolumn{4}{c|}{\multirow{4}[2]{*}{(2.5,2,1.5,0.5)}} & 0.6248 & 0.4635 & 0.3663 & 0.0937 & 0.5612 & 0.2966 & 0.1811 & 0.0130 \\
			20    & \multicolumn{4}{c|}{}         & 0.5938 & 0.3938 & 0.3094 & 0.0898 & 0.5361 & 0.2423 & 0.1371 & 0.0115 \\
			30    & \multicolumn{4}{c|}{}         & 0.5493 & 0.3932 & 0.2783 & 0.0881 & 0.5016 & 0.2355 & 0.1115 & 0.0112 \\
			40    & \multicolumn{4}{c|}{}         & 0.3742 & 0.3755 & 0.2570 & 0.0870 & 0.2841 & 0.2310 & 0.1002 & 0.0112 \\
			\midrule
			\multicolumn{13}{c}{$c=1$} \\
			\midrule
			10    & \multicolumn{4}{c|}{\multirow{4}[2]{*}{(2,1.5,1.5,0.5)}} & 0.4990 & 0.3431 & 0.4215 & 0.0981 & 0.3891 & 0.1711 & 0.2275 & 0.0150 \\
			20    & \multicolumn{4}{c|}{}         & 0.4855 & 0.3107 & 0.3314 & 0.0923 & 0.3810 & 0.1584 & 0.1554 & 0.0126 \\
			30    & \multicolumn{4}{c|}{}         & 0.4809 & 0.3004 & 0.2934 & 0.0861 & 0.3518 & 0.1481 & 0.1230 & 0.0112 \\
			40    & \multicolumn{4}{c|}{}         & 0.4069 & 0.2772 & 0.2703 & 0.0812 & 0.3366 & 0.1278 & 0.1076 & 0.0099 \\
			\midrule    10    & \multicolumn{4}{c|}{\multirow{4}[2]{*}{(2,2,1.5,0.5)}} & 0.4901 & 0.5046 & 0.4897 & 0.1017 & 0.3628 & 0.3482 & 0.2926 & 0.0156 \\
			20    & \multicolumn{4}{c|}{}         & 0.4782 & 0.3706 & 0.3686 & 0.0968 & 0.3518 & 0.2177 & 0.1834 & 0.0141 \\
			30    & \multicolumn{4}{c|}{}         & 0.4532 & 0.3683 & 0.3205 & 0.0903 & 0.3291 & 0.2116 & 0.1443 & 0.0120 \\
			40    & \multicolumn{4}{c|}{}         & 0.3678 & 0.3558 & 0.2841 & 0.0835 & 0.2908 & 0.2093 & 0.1157 & 0.0103 \\
			\midrule    10    & \multicolumn{4}{c|}{\multirow{4}[2]{*}{(2.5,2,1.5,0.5)}} & 0.6137 & 0.5036 & 0.4152 & 0.0975 & 0.5312 & 0.3414 & 0.2234 & 0.0138 \\
			20    & \multicolumn{4}{c|}{}         & 0.5798 & 0.3969 & 0.3305 & 0.0928 & 0.5192 & 0.2372 & 0.1537 & 0.0127 \\
			30    & \multicolumn{4}{c|}{}         & 0.5308 & 0.3884 & 0.2903 & 0.0918 & 0.4820 & 0.2337 & 0.1198 & 0.0119 \\
			40    & \multicolumn{4}{c|}{}         & 0.3601 & 0.3702 & 0.2655 & 0.0909 & 0.2633 & 0.2216 & 0.1051 & 0.0118 \\
			\bottomrule    \end{tabular}%
		
		\label{tab:mcmc_2}%
	\end{table}%
\end{landscape}
\begin{landscape}
	\begin{table}[h]
		\centering
		\caption{Bias and mean squared error (MSE) of maximum likelihood estimators.}
		
		\begin{tabular}{c|cccc|cccccccc}
			\toprule
			\multirow{2}[4]{*}{$n$} & \multicolumn{4}{c|}{\multirow{2}[4]{*}{$(a,b_1,b_2,\theta)$}} & \multicolumn{4}{c}{Bias}      & \multicolumn{4}{c}{MSE} \\
			\cmidrule{6-13}          & \multicolumn{4}{c|}{}         & $a$   & $b_1$ & $b_2$ & $\theta$ & $a$   & $b_1$ & $b_2$ & $\theta$ \\
			\midrule
			10    & \multicolumn{4}{c|}{\multirow{4}[1]{*}{(2,1.5,1.5,0.5)}} & 0.9675 & 0.4735 & 0.3771 & 0.0559 & 0.6315 & 0.4462 & 0.4651 & 0.0470 \\
			20    & \multicolumn{4}{c|}{}         & 0.7561 & 0.3918 & 0.1813 & 0.0541 & 0.5806 & 0.3536 & 0.3117 & 0.0225 \\
			30    & \multicolumn{4}{c|}{}         & 0.6650 & 0.3127 & 0.1042 & 0.0513 & 0.5284 & 0.3046 & 0.2285 & 0.0159 \\
			40    & \multicolumn{4}{c|}{}         & 0.6263 & 0.2607 & 0.0678 & 0.0309 & 0.4770 & 0.2494 & 0.1809 & 0.0127 \\
			\midrule
			10    & \multicolumn{4}{c|}{\multirow{4}[0]{*}{(2,2,1.5,0.5)}} & 0.9681 & 0.5318 & 0.3779 & 0.0559 & 0.6314 & 0.3966 & 0.4664 & 0.0470 \\
			20    & \multicolumn{4}{c|}{}         & 0.7562 & 0.6157 & 0.1820 & 0.0541 & 0.5806 & 0.2861 & 0.3131 & 0.0225 \\
			30    & \multicolumn{4}{c|}{}         & 0.6650 & 0.5619 & 0.1042 & 0.0513 & 0.5284 & 0.2407 & 0.2284 & 0.0159 \\
			40    & \multicolumn{4}{c|}{}         & 0.6263 & 0.5278 & 0.0678 & 0.0310 & 0.4770 & 0.2300 & 0.1809 & 0.0127 \\
			\midrule
			10    & \multicolumn{4}{c|}{\multirow{4}[0]{*}{(2.5,1.5,1.5,0.5)}} & 0.9880 & 0.4605 & 0.4617 & 0.0662 & 0.4001 & 0.4383 & 0.4563 & 0.0480 \\
			20    & \multicolumn{4}{c|}{}         & 0.9259 & 0.3527 & 0.2681 & 0.0643 & 0.3788 & 0.3468 & 0.3437 & 0.0253 \\
			30    & \multicolumn{4}{c|}{}         & 0.8561 & 0.2724 & 0.1863 & 0.0608 & 0.3716 & 0.2935 & 0.2549 & 0.0176 \\
			40    & \multicolumn{4}{c|}{}         & 0.8109 & 0.2241 & 0.1477 & 0.0439 & 0.3395 & 0.2395 & 0.2036 & 0.0144 \\
			\midrule
			10    & \multicolumn{4}{c|}{\multirow{4}[0]{*}{(2.5,2,1.5,0.5)}} & 0.9867 & 0.4779 & 0.4639 & 0.0662 & 0.4001 & 0.3955 & 0.4555 & 0.0480 \\
			20    & \multicolumn{4}{c|}{}         & 0.9256 & 0.5408 & 0.2680 & 0.0643 & 0.3789 & 0.2767 & 0.3453 & 0.0253 \\
			30    & \multicolumn{4}{c|}{}         & 0.8562 & 0.5332 & 0.1864 & 0.0608 & 0.3711 & 0.2495 & 0.2552 & 0.0176 \\
			40    & \multicolumn{4}{c|}{}         & 0.8109 & 0.4979 & 0.1476 & 0.0440 & 0.3381 & 0.2211 & 0.2038 & 0.0144 \\
			\bottomrule
		\end{tabular}%
		
		\label{tab:mle}%
	\end{table}%

\end{landscape}
\begin{figure}[h]
	\centering
	\begin{minipage}{.5\textwidth}
		\centering
		\includegraphics[width=1\linewidth]{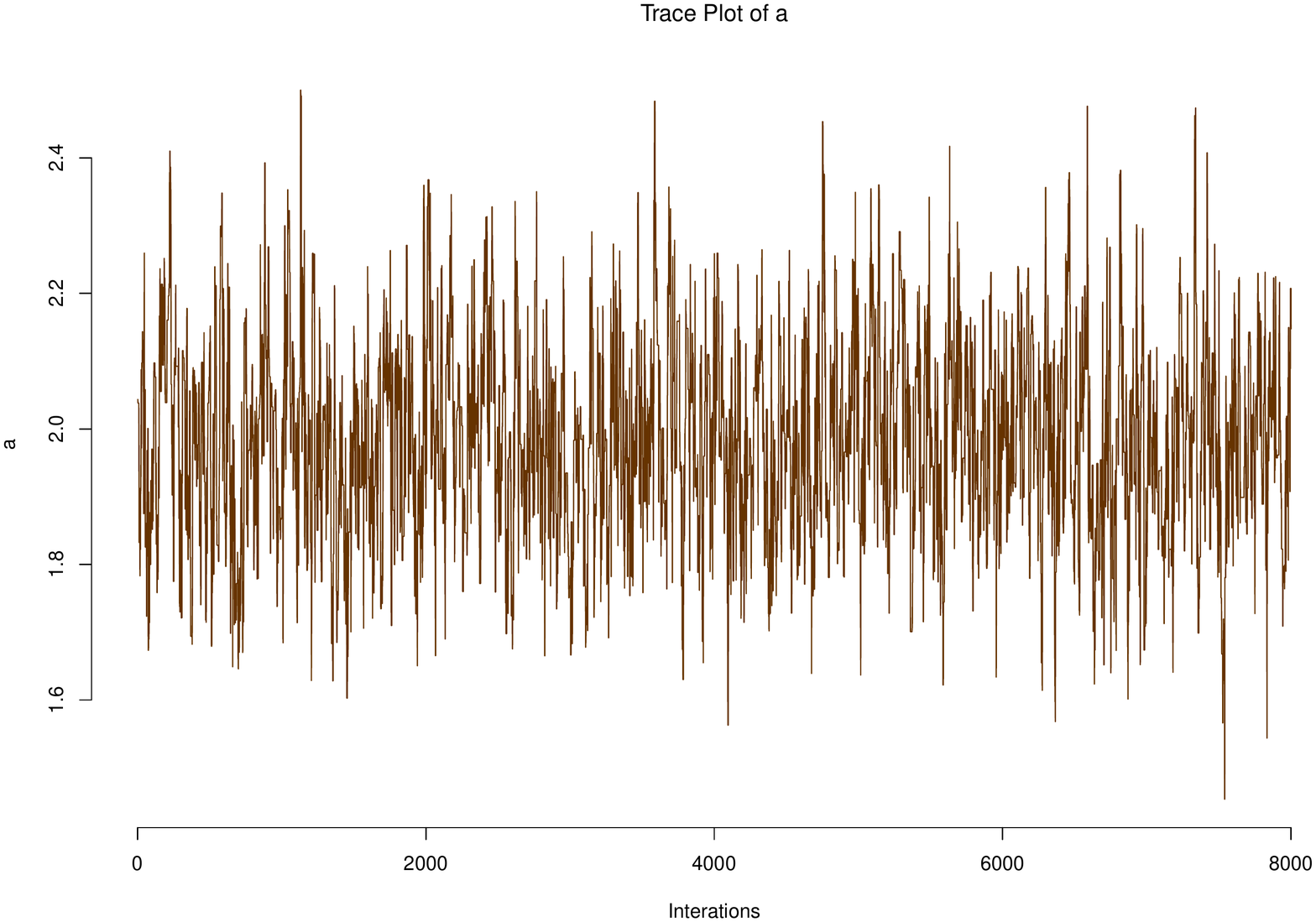}
		\captionof{figure}{Trace Plot of $a$}
		\label{fig:tracea}
	\end{minipage}%
	\begin{minipage}{.5\textwidth}
		\centering
		\includegraphics[width=1\linewidth]{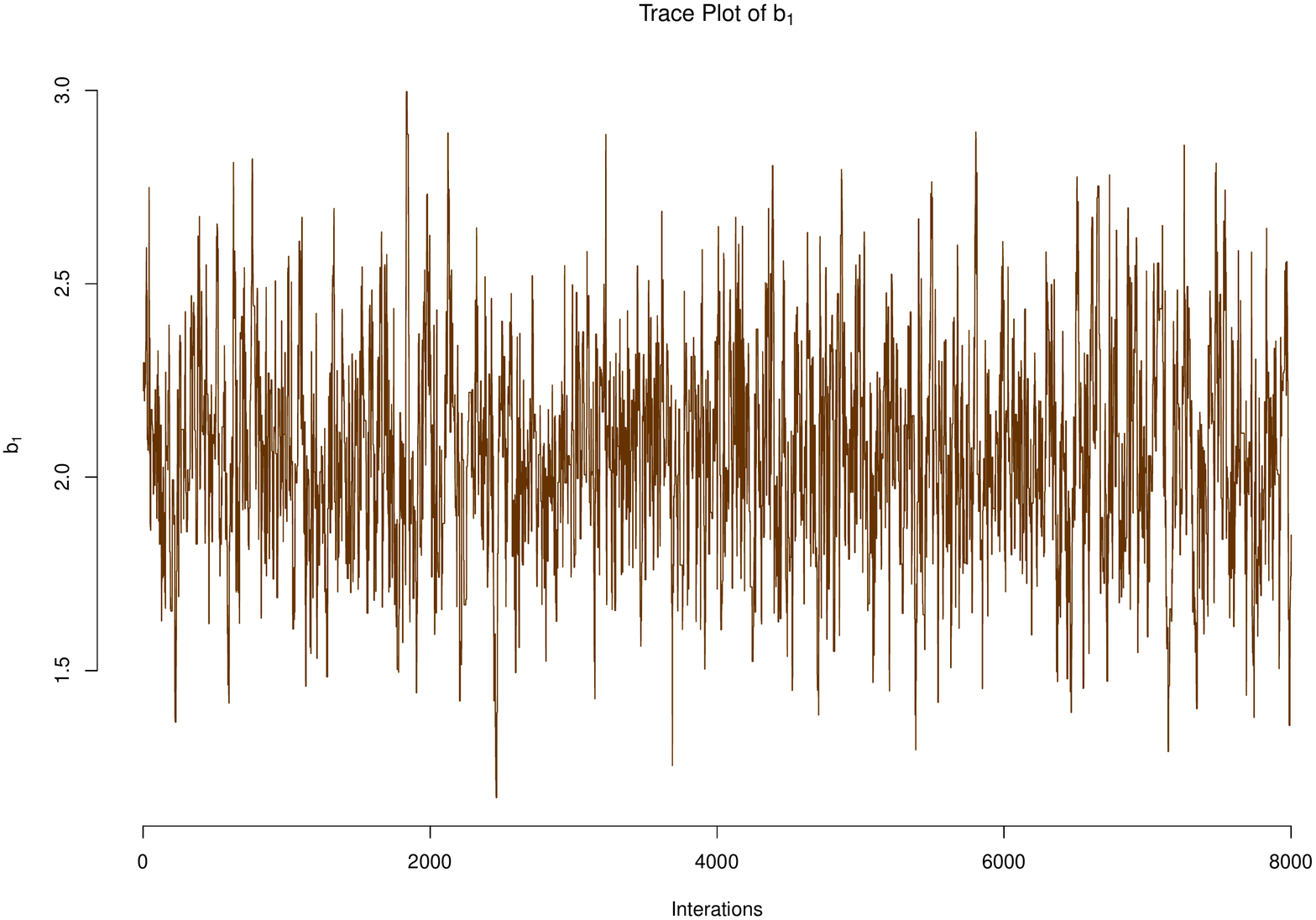}
		\captionof{figure}{Trace Plot of $b_1$}
		\label{fig:traceb1}
	\end{minipage}
\end{figure}
\begin{figure}[h]
	\centering
	\begin{minipage}{.5\textwidth}
		\centering
		\includegraphics[width=1\linewidth]{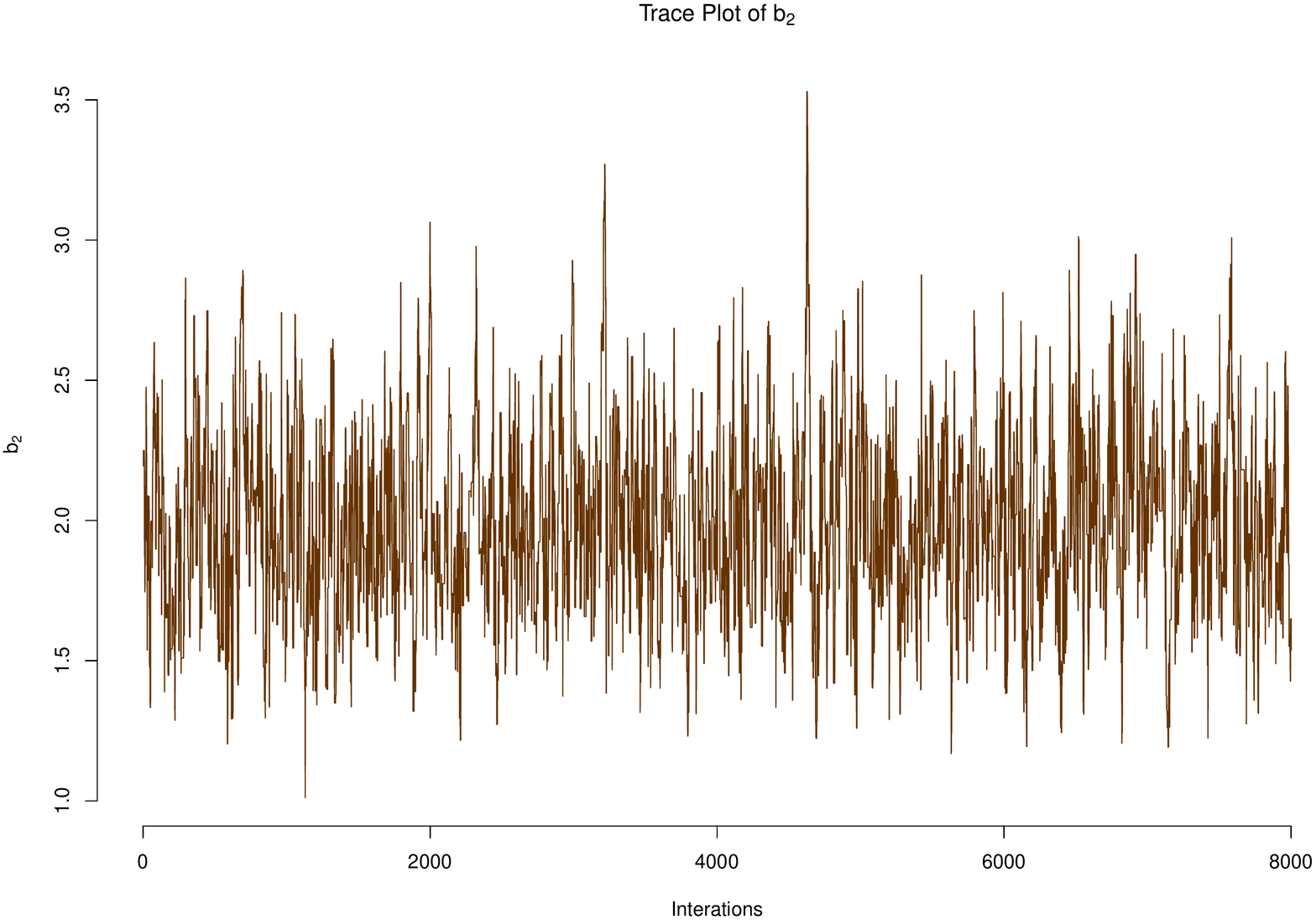}
		\captionof{figure}{Trace Plot of $b_2$}
		\label{fig:traceb2}
	\end{minipage}%
	\begin{minipage}{.5\textwidth}
		\centering
		\includegraphics[width=1\linewidth]{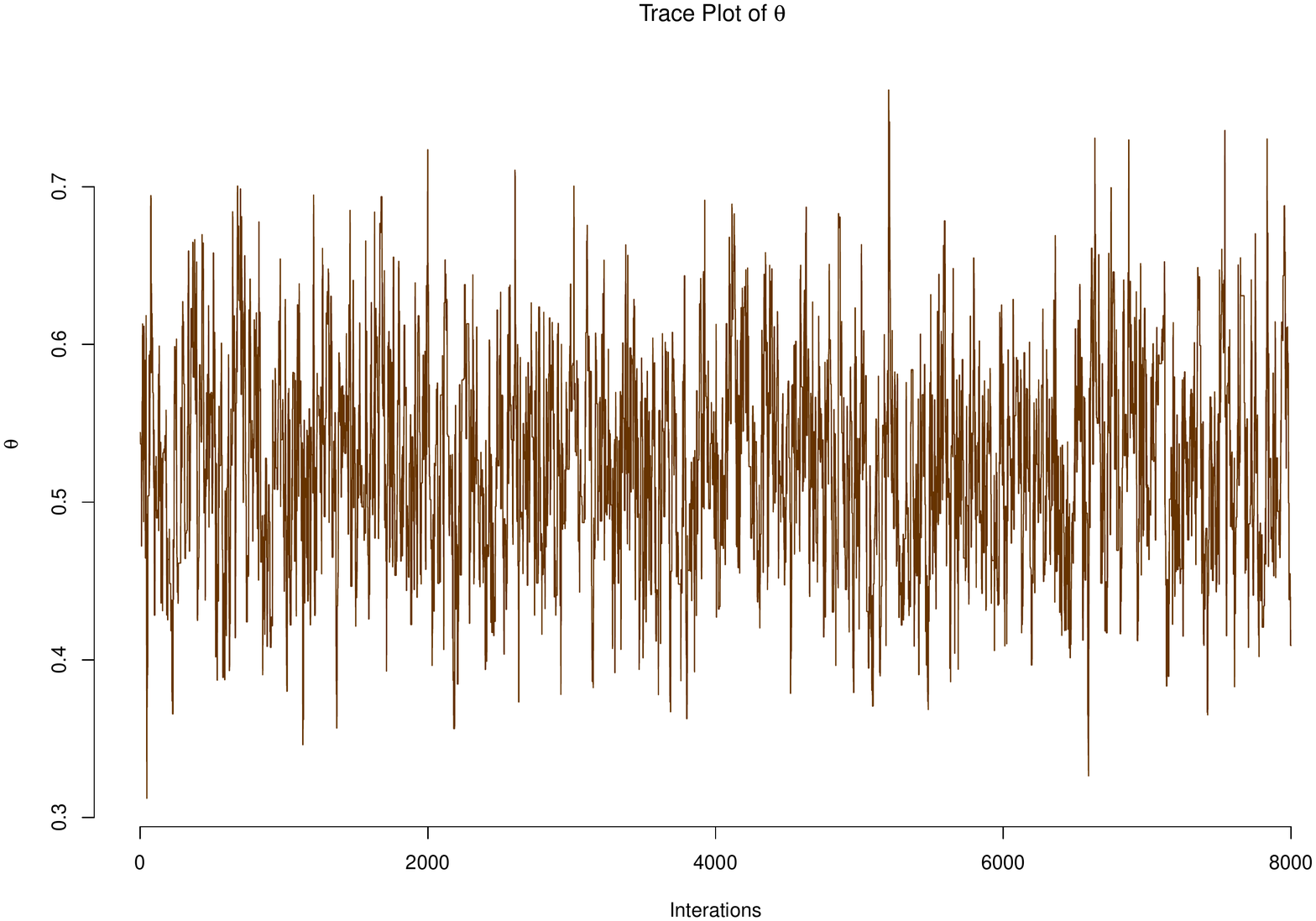}
		\captionof{figure}{Trace of $\theta$}
		\label{fig:tracetheta}
	\end{minipage}
\end{figure}
\section{A real application}
In this section, we consider the American Football (National Football League) League data set, reported in Jamalizadeh and Kundu
	(2013). In this data, the variable $X$ represents the game time to the first points
	scored by kicking the ball between goal posts and $Y$ represents the `game time' by moving the
	ball into the end zone. We first calculate descriptive statistics and some basic measures of dependence, namely Pearsons's correlation coefficient, Spearman's rho, Kendall's tau, 	Blest's measure and Spearman's footrule coefficient for the considered data set. The values of these quantities are reported in Table \ref{discriptive}. The calculated values of Pearsons's correlation coefficient, Spearman's rho, Kendall's tau, Blest's measure and Spearman's footrule coefficient are 0.7226, 0.8038, 0.6802, 0.6171 and 0.8276, respectively, which clearly exhibits a positive associative in considered data. 
	\begin{table}[!ht]
		\centering
			\caption{Descriptive statistics and measures of dependence of the American Football data.}
		\begin{tabular}{@{}lccc@{}}
			\toprule
			\multicolumn{1}{c}{{\begin{tabular}[c]{@{}c@{}}Statistics{}\end{tabular}}} & {\begin{tabular}[c]{@{}c@{}}X\end{tabular}} & {\begin{tabular}[c]{@{}c@{}}{}\end{tabular}} & {Y}\\ \midrule
			Minimum & 0.7500 &\hspace*{2cm} & 0.7500\\
			Maximum & 32.4500 &\hspace*{2cm} & 49.8800\\
			1st Quantile & 4.2280 &\hspace*{2cm} & 6.4230\\
			Mean & 9.0740 &\hspace*{2cm} & 13.4250\\
			Median & 7.5150 &\hspace*{2cm} & 9.9150\\
			3rd Quantile & 11.4350 &\hspace*{2cm} & 14.9550\\
			Skewness & 1.6664 &\hspace*{2cm} & 1.6750\\
			Kurtosis & 6.3692 &\hspace*{2cm} & 5.1236\\
			Standard deviation & 6.8359 &\hspace*{2cm} & 12.3285\\
			\hline
			Pearson's correlation & \hspace*{2cm} & 0.7226& \hspace*{2cm} \\
			Spearman's rho & \hspace*{2cm} & 0.8038& \hspace*{2cm} \\
			Kendall's tau & \hspace*{2cm} & 0.6802& \hspace*{2cm} \\
			Spearman's footrule coeff. &  & 0.6171& \\ 
			Blest's measure &  & 0.8276& \\ \bottomrule
		\end{tabular}
	\label{discriptive}
	\end{table}

\noindent To show the applicability of the result, we have to check whether the dataset $X$ and $Y$ support assumed families of distributions. For this purpose, we consider Kolmogorov-Smirnov (KS) test and find out that $X$ supports exponentiated Weibull distribution for $a=1.0606,$ $b_1= 0.9958$ and $\theta=0.9999$  with $p$-value as 0.4765 and KS distance as 0.1301. In a similar manner, we find that $Y$ supports exponentiated Weibull distribution for $ a=1.2429$ , $b_1= 0.8123,$ and $\theta=0.9983$ with $p$-value as 0.3802 and KS distance as 0.13646.
These results can easily be visualized graphically in Figure [\ref{fig:fite_cdf_X}-\ref{fig:fite_cdf_Y}]. 
Now, after discussing the fitting of marginals to this data. We consider the fitting for BGW distribution and compare the proposed model from submodels of  BGW distribution. The considered submodels are bivariate generalized Rayleigh (BGR) distribution and bivariate generalized exponential (BGE) distribution.  The considered data set is used by Jamalizadeh and Kundu
(2013) to show the application of their proposed weighted Marshall-Olkin bivariate exponential distribution (WMOBE)  with the Marshall-Olkin bivariate Weibull (MOBW) (Jamalizadeh and Kundu
(2013)) distribution. The authors concluded that WMOBE provides a better fit over the MOBW. In this article, we consider the same dataset to show that our proposed model provides a better fit over WMOBE and MOBW distribution.  The comparison  is made based on the log-likelihood function, Akaike information criteria (AIC), and Bayesian information criteria (BIC). The values of AIC is calculated by $2p-2\ln L$, and BIC is calculated by $p\ln n-2\ln L,$ where $p$ is the number of parameters, $n$ is the number of observations and $L$ is the maximum value of the likelihood. Table [\ref{tab:par_est}] presents estimates and other quantities of the data with respect to models. From this table, we infer that BGW distribution has the minimum value of the AIC and BIC, and maximum value of log-likelihood. So, with respect to these findings, we conclude that the considered dataset supports BGW distribution best among other distributions. The estimates of the unknown parameters given in Table [\ref{tab:par_est}] are MLEs. Now we calculate the MCMC estimates of the parameters of BGW distribution. The results are calculated and reported in Table [\ref{tab:real_mcmc}].

\begin{figure}[h]
	\centering
	\begin{minipage}{.5\textwidth}
		\centering
		\includegraphics[width=1\linewidth]{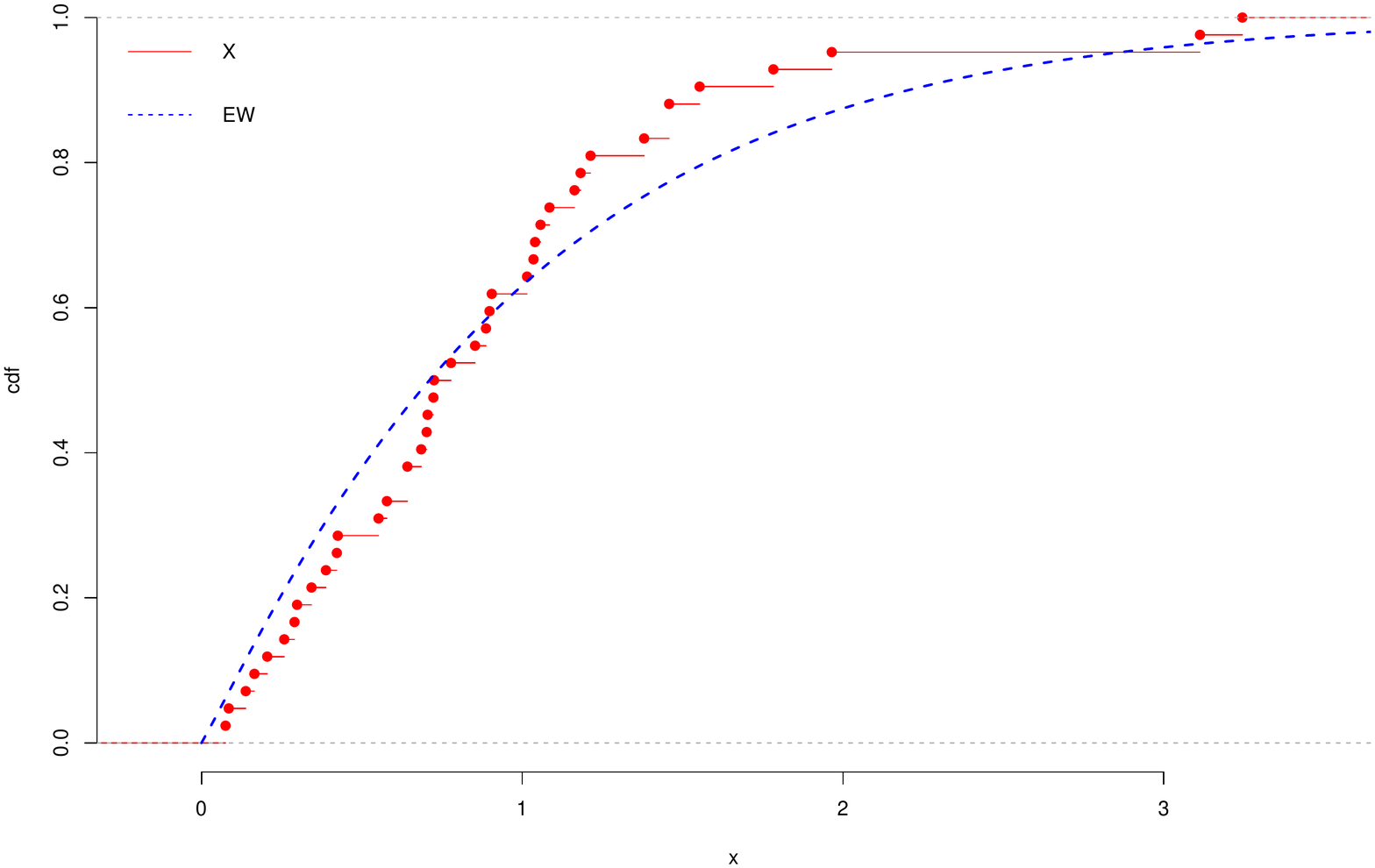}
		\captionof{figure}{Fitted CDF Plot of $X$}
		\label{fig:fite_cdf_X}
	\end{minipage}%
	\begin{minipage}{.5\textwidth}
		\centering
		\includegraphics[width=1\linewidth]{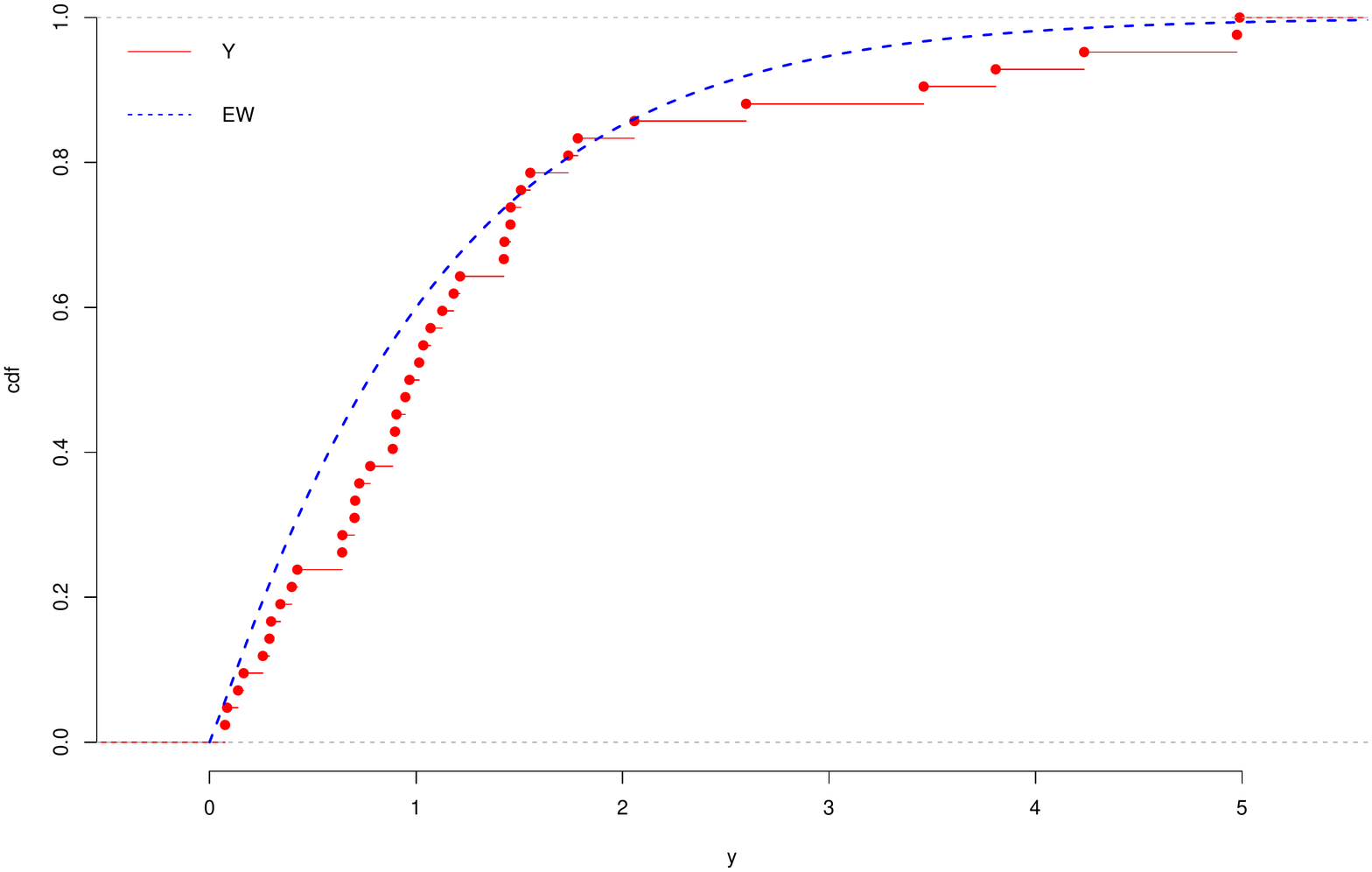}
		\captionof{figure}{Fitted CDF Plot of $Y$}
		\label{fig:fite_cdf_Y}
	\end{minipage}
\end{figure}

\begin{table}[htbp]
	\centering
	\caption{Parameter estimates with log-likelihood and AIC values}
	
	\begin{tabular}{l|cccccc|c|cc}
		\toprule
		Distribution & \multicolumn{6}{c|}{Estimates}                & LL    & AIC   & BIC \\
		\midrule	BGW   & \multicolumn{6}{c|}{$a=3.9834,~b_1 = 0.0150,~b_2=  0.0059,~\theta=   0.1701$} & -65.2343 & 138.4686 & 145.4193 \\
		BGE   & \multicolumn{6}{c|}{$b_1=1.0331,~b_2=   0.7082,~\theta =  0.9160$} & -91.9866 & 189.9732 & 198.5239 \\
		BGR   & \multicolumn{6}{c|}{$b_1=0.3161,~b_2=   0.1662,\theta=   0.4071$} & -73.0684 & 152.1368 & 157.3498 \\
		WMOBE & \multicolumn{6}{c|}{$\lambda_1=0.5996,~\lambda_2=0.0346,~\theta=0.8639,~\alpha=2.5302$} & -85.4447 & 178.8894 & 185.8401 \\
		MOBW  & \multicolumn{6}{c|}{$\lambda_1=1.2889,~\lambda_2=0.5761,~\theta=0.4297,~\alpha=0.0244$} & -90.4169 & 188.8338 & 195.7845 \\
		\hline
	\end{tabular}%
	
	\label{tab:par_est}%
\end{table}%
\begin{table}[htbp]
	\centering
	\caption{MCMC Bayes estimates of parameters for Prior $(\delta_i,\zeta_i)$, $i=1,2,3,4$}
	\begin{tabular}{c|c|cccc}
		\toprule
		\multirow{2}[4]{*}{Prior} & \multirow{2}[4]{*}{$c$} &       & MCMC Bayes Estimates &       &  \\
		\cmidrule{3-6}          &       & $a$   & $b_1$ & $b_2$ & $\theta$ \\
		\midrule
		& -1    & 2.5649 & 0.1864 & 0.0882 & 0.3235 \\
		& -0.5  & 2.5511 & 0.1709 & 0.0798 & 0.3184 \\
		$(1.5,1.5)$ & 0.5   & 2.523 & 0.1424 & 0.0647 & 0.3087 \\
		& 1     & 2.5087 & 0.1298 & 0.0582 & 0.304 \\
		& 1.5   & 2.4943 & 0.1184 & 0.0524 & 0.2995 \\
		\cmidrule{2-6}          & -1    & 2.3968 & 0.2369 & 0.1164 & 0.3551 \\
		& -0.5  & 2.3827 & 0.2209 & 0.1072 & 0.3501 \\
		$(2,2)$ & 0.5   & 2.3551 & 0.1855 & 0.0872 & 0.3403 \\
		& 1     & 2.3416 & 0.1659 & 0.076 & 0.3355 \\
		& 1.5   & 2.3284 & 0.1457 & 0.0642 & 0.3307 \\
		\bottomrule
	\end{tabular}%
	\label{tab:real_mcmc}%
\end{table}%
\noindent Since $\text{BGE}(b_1, b_2, \theta)$ and $\text{BGR}(b_1, b_2, \theta)$ reported in (\ref{dist11}) and (\ref{RD1}), respectively, are sub-models of $\text{BGW}(a, b_1, b_2, \theta)$. We consider the test of the following hypothesis:\\
(i) $H_0 : a=1~ \text{(BGE})$ against $H_1 : a \neq 1 ~\text{(BGW})$ and (ii) $H_0^{*} : a=2~ \text{(BGR})$ against $H_1^{*} : a \neq 2 ~\text{(BGW})$ and carry out the likelihood ratio tests.  The log-likelihood ratio test statistic value for (i) hypothesis is $-2[\ln L_{BGE}-\ln L_{BGW}]=53.5046$ with corresponding $p$-value approximately zero. Further, for (ii) hypothesis, $-2[\ln L_{BGR}-\ln L_{BGW}]=15.6682$ with corresponding $p$-value 0.00007. Considering the values of test statistic and associated $p$-values, we conclude that BGW distribution provides a better fit over the BGE and BGR distribution for the considered data set. 
\section*{Conclusions}
\noindent This article presents a novel absolute continuous bivariate generalized Weibull (BGW) distribution. The univariate marginals of this distribution are exponentiated Weibull distributions. The proposed model has bivariate generalized exponential (BGE) (see Mirhosseini {\it et al.} (2015)) and bivariate generalized Rayleigh distribution (see Pathak and Vellaisamy (2021)) as sub-models for specific values of parameters. Several properties of the BGW distribution are presented such as distribution function, survival function, density function etc. Results pertaining to product moments of the distribution are given which are further reduced for the sub-models of the distributions. For reliability and lifetime analysis, the notion of dependence is discussed with the aid of positive quadrant dependence, regression dependence, stochastic increasing, totally positivity of order 2, etc. Apart from that, various dependence measures are provided for the BGW model {\it e.g.,} copula based dependence, tail coefficient dependence and regression dependence. The authors have also considered estimation of unknown parameters under classical and Bayesian paradigm. For the computational part, a rigorous simulation study is conducted to observe the behaviour of estimates of the parameters using mean squared error criteria. Finally, we have also shown that BGW distribution works well in real data application.
 \section{Appendix}
 \begin{proof}[\bf Proof of Theorem {\ref{regproposition}}]
 	We have
 	\begin{align}\label{appendixreg1}
 	E(Y|X=x)&=\int_{0}^{\infty}y f(y|x)dy
 	=\frac{1}{f_{X}(x)}\int_{0}^{\infty}y f(x,y)dy.
 	\end{align}
 	Using (\ref{dist3}) in (\ref{appendixreg1}), we get
 	\begin{align*}
 	E(Y|X=x)&=\frac{1}{f_{X}(x)}\int_{0}^{\infty}a^2b_1b_2x^{a-1}y^{a}\sum_{j=1}^{\infty}\binom{\theta}{j}(-1)^{j+1} j^{2} 
 	e^{-j\left(b_{1}x^{a}+b_{2}y^{a}\right)}dy.
 	\end{align*}
 	Due to absolute integrability of the summand, we can interchange summation and integration. Hence, we get
 	\begin{align*}\label{appendixreg2}
 	E(Y|X=x)&=\frac{ab_{1}x^{a-1}}{f_{X}(x)} \sum_{j=1}^{\infty}\binom{\theta}{j}(-1)^{j+1} j^{2} e^{-jb_{1}x^{a}}\left\{\int_{0}^{\infty} ab_2 y^{a}e^{-jb_{2}y^{a}}dy\right\}.
 	\end{align*}
 	Evaluation of integral inside the bracket completes the proof of Theorem \ref{regproposition}.
 \end{proof}
  
 \begin{proof}[\bf Proof of Theorem \ref{Momenttheorem}]
	Product moment in terms of density is defined as
	\begin{equation}\label{appendix1}
	E(X^{r}Y^{s})=\int_{0}^{\infty}\int_{0}^{\infty}x^{r} y^{s} f(x,y) dx dy.
	\end{equation}
	Putting $f(x,y)$ from (\ref{dist3}) in (\ref{appendix1}), we get
	\begin{equation}
	E(X^{r}Y^{s})=\int_{0}^{\infty}\int_{0}^{\infty} x^{r} y^{s} a^2b_{1}b_{2}x^{a-1}y^{b-1}\sum_{j=1}^{\infty}\binom{\theta}{j}(-1)^{j+1} j^{2} e^{-j\left(b_{1}x^{a}+b_{2}y^{a}\right)}dx dy.
	\end{equation}
	Due to absolute integrability of the summand, we can interchange summation and integration. Therefore
	\begin{equation*}\label{appendix2}
	E(X^{r}Y^{s})=\sum_{j=1}^{\infty}\binom{\theta}{j}(-1)^{j+1} j^{2} ~L_{1} L_{2},
	\end{equation*}
	where  
	\begin{equation*}\label{appendix21}
	L_{1}=\int_{0}^{\infty} \displaystyle ab_{1}x^{r+a-1}e^{-j{b_{1}}x^{a}} dx=\displaystyle\frac{\Gamma(1+r/a)}{j(b_{1}j)^{r/a}}
	\end{equation*}
	and
	\begin{equation*}\label{appendix22}
	L_{2}=\int_{0}^{\infty} \displaystyle ab_{2}y^{s+a-1}e^{-j{b_{2}}y^{a}} dy=\displaystyle\frac{\Gamma(1+s/a)}{j(b_{2}j)^{s/a}}.
	\end{equation*}
	Hence the proof complete.
\end{proof}

 \begin{proof}[\bf Proof of Theorem \ref{orderstat}]\hfil{}\\
	(i)		Since
	\begin{align*}
	P(\min\{X,Y\}> s)&=P(X > s, Y > s)\nonumber\\
	&=\sum_{k=1}^{\infty}\left[P(U_{i} > s) P(V_{i} > s)\right]^{k}P(K=k)\nonumber\\
	&=h_{K}\left(e^{-\left(b_1 t^{a}+b_2t^{a}\right)}\right) \;\;\;\;\;\;\;\hfill{(\text{Using~ eq~} (\ref{PG1}))}\nonumber\\
	&=1-\left\{1-e^{-(b_{1}+b_{2})t^{a}}\right\}^{\theta},
	\end{align*}
	we have
	\begin{equation*}
	P(\min\{X,Y\}\leq s)=\left\{1-e^{-(b_{1}+b_{2})s^{a}}\right\}^{\theta},
	\end{equation*}
	which establish the first part of the theorem.\\
	(ii) 	We have
	\begin{equation}\label{appendixstress1}
	P(X<Y)=\int_{0}^{\infty}\int_{0}^{y} f(x,y) dx dy.
	\end{equation}
	Using (\ref{dist3}) in (\ref{appendixstress1}), we get
	\begin{align}\label{appendixstress2}
	P(X<Y)&=\int_{0}^{\infty}\int_{0}^{y}a^2b_{1}b_{2}x^{a-1}y^{b-1}\sum_{j=1}^{\infty}\binom{\theta}{j}(-1)^{j+1} j^{2} e^{-j\left(b_{1}x^{a}+b_{2}y^{a}\right)}dx dy\nonumber\\
	&=\sum_{j=1}^{\infty}\binom{\theta}{j}(-1)^{j+1} j^{2}\int_{0}^{\infty} ab_{2}y^{a-1}e^{-j {b_{2}y^{a}}}\left\{\int_{0}^{y}a b_{1}x^{a-1}
	e^{-jb_{1}x^{a}}dx\right\}dy\nonumber\\
	&=1-\sum_{j=1}^{\infty}\binom{\theta}{j}(-1)^{j+1} j \int_{0}^{\infty}a b_{2}y^{a-1}e^{-\displaystyle j(b_{1}+b_{2})y^{a}} dy
	\end{align}
	Integration of (\ref{appendixstress2}) and bit algebra leads to proof of the result.
\end{proof}

\begin{proof}[\bf Proof of Theorem {\ref{regressiondependence}}] Differentiating equation (\ref{copula3}) partially with respect to $s$, we get
	\begin{equation}\label{regdepp1}
\frac{\partial}{\partial s}C(s,t)=	1-s^{\frac{1}{\theta}-1}\left(1-t^{\frac{1}{\theta}}\right)\left\{1-\left(1-s^{\frac{1}{\theta}}\right)\left(1-t^{\frac{1}{\theta}}\right)\right\}^{\theta-1}.
	\end{equation}
	Taking square of (\ref{regdepp1}), we get
	\begin{align}\label{regdeprevise}
	\left(\frac{\partial}{\partial s}C(s,t)\right)^2=&1+s^{2\left(\frac{1}{\theta}-1\right)})\left(1-t^{\frac{1}{\theta}}\right)^2\left\{1-\left(1-s^{\frac{1}{\theta}}\right)\left(1-t^{\frac{1}{\theta}}\right)\right\}^{2(\theta-1)}\nonumber
	\\
	&-2s^{\frac{1}{\theta}-1}\left(1-t^{\frac{1}{\theta}}\right)	\left\{1-\left(1-s^{\frac{1}{\theta}}\right)\left(1-t^{\frac{1}{\theta}}\right)\right\}^{\theta-1}\end{align}
The binomial series expansion of (\ref{regdeprevise}) leads to
	\begin{align}\label{regdepp2}
	\left(\frac{\partial}{\partial s}C(s,t)\right)^2=&1+\sum_{j=0}^{\infty}\binom{\theta-1}{j}(-1)^{j}s^{2\left(\frac{1}{\theta}-1\right)}\left(1-s^{\frac{1}{\theta}}\right)^{j}\left(1-t^{\frac{1}{\theta}}\right)^{j+2}\nonumber\\
	&-2\sum_{j=0}^{\infty}\binom{\theta-1}{j}(-1)^{j}s^{\frac{1}{\theta}-1}\left(1-s^{\frac{1}{\theta}}\right)^{j}\left(1-t^{\frac{1}{\theta}}\right)^{j+1}.
\end{align}
Putting (\ref{regdepp2}) in (\ref{regdep0}) and after integrating with respect to $s$ and $t$, we get the proof of the theorem.
\end{proof}


\end{document}